\documentclass{amsart}
\usepackage{amsfonts,amsmath,graphicx}
\graphicspath{{./figures/}}
\usepackage{pdfsync}
\usepackage{subfigure}
\usepackage{color}
\newtheorem{theorem}{Theorem}
\theoremstyle{plain}

\def\be{\begin{equation}}
\def\ee{\end{equation}}

\newtheorem{lemma}{Lemma}
\newtheorem{notation}{Notation}

\newtheorem{remark}{Remark}
\numberwithin{equation}{section}
\renewcommand\Re{\operatorname{Re}}

\def\Pr{\text{Pr}}

\begin{document}
\title[Pattern Formation in Rayleigh B{\'e}nard Convection]{Pattern Formation in Rayleigh--B{\'e}nard Convection}
\author[Sengul]{Taylan Sengul}
\address[Sengul]{Department of Mathematics,
Indiana University, Bloomington, IN 47405}
\email{msengul@indiana.edu}

\author[Wang]{Shouhong Wang}
\address[Wang]{Department of Mathematics,
Indiana University, Bloomington, IN 47405}
\email{showang@indiana.edu, http://www.indiana.edu/\~{ }fluid}
\thanks{The work  was supported in part by the
Office of Naval Research and by the National Science Foundation.}
\begin{abstract}
The main objective of this article is to study the three-dimensional Rayleigh-B{\'e}nard convection in a rectangular domain from a pattern formation perspective. It is well known that as the Rayleigh number crosses a critical threshold, the system undergoes a Type-I transition, characterized by an attractor bifurcation. The bifurcated attractor is an $(m-1)$--dimensional homological sphere where $m$ is the multiplicity of the first critical eigenvalue. When $m=1$, the structure of this attractor is trivial. When $m=2$, it is known that the bifurcated attractor consists of  steady states and their connecting heteroclinic orbits. The main focus of this article is then on the pattern selection mechanism and stability of rolls, rectangles and mixed modes (including hexagons) for the case where $m=2$. We derive in particular  a complete  classification of  all transition scenarios, determining the patterns of the bifurcated steady states, their stabilities and the basin of attraction of the stable ones. The theoretical results lead to interesting physical conclusions, which are in agreement with known experimental results. For example, it is shown in this article that  only the pure modes are stable whereas the mixed modes are unstable.
\end{abstract}

\maketitle
\section{Introduction}
Over the years,  the Rayleigh-B\'enard convection problem, together with the Taylor problem, has
become one of the paradigms for studying nonequilibrium phase transitions and pattern formation in nonlinear sciences. There is an extensive literature on the subject; see e.g.  reviews by Busse \cite{busse1978non}, Chandrasekhar \cite{Chandrasekhar1981}, Cross \cite{Cross1993}, Getling \cite{Getling1997},
Koschmieder \cite{Koschmieder1993}, Lappa \cite{Lappa2009}, Ma and Wang \cite{ptd}, and the references therein.

The problem is complete from the dynamic transition perspective (Ma and Wang \cite{Ma2004a,Ma2007}). The main result in this direction is that the system always undergoes a Type-I (continuous) transition as the instability driving mechanism, namely Rayleigh number, crosses a critical threshold $R_c$, thanks to the symmetry of the linear operator, properties of the nonlinearity and asymptotic stability of the basic state at the critical threshold. Moreover, the system has a bifurcated attractor which is an $(m-1)$--dimensional homological sphere where $m$ is the number of critical eigenvalues of the linear operator. 

The main objective of this paper is to study pattern formation and the structure of the bifurcated attractor for the Rayleigh--B{\'e}nard convection.
The structure of the bifurcated attractor is trivial when $m=1$. Namely, the attractor consists of two attracting steady states approximated by the critical mode  with opposite flow orientations. 

When $m=2$, the picture is far from being complete. There are some known characteristics of this attractor such as the attractor must be homeomorphic to $S^1$ which contains either four or eight steady states connected by heteroclinic orbits or is a circle of steady states.

From a pattern formation point of view, there is enough motivation to study the case $m=2$. To understand the relative  stabilities of steady states (patterns) to perturbations of other pattern types, there must exist at least two critical modes. 

In general the relation between the two horizontal length scales, for which $m=2$, is nonlinear and hence it is very difficult to give a general characterization of every possible transition scenario. In this work, under the assumption that the wave numbers of the critical modes are equal, we are able to give a complete characterization for $m=2$ case.  

Depending on its horizontal wave indices $i_x$ and $i_y$, a single critical mode can be either a roll (when at least one of $i_x$ or $i_y$ is zero) or a rectangle (when both $i_x$ and $i_y$ are non-zero) where $i_x$ and $i_y$ are non-negative integers which cannot vanish together. Thus there are three possible cases. Namely, (a) one of the critical modes is a roll while the other one is a rectangle, (b) both critical modes are rolls, (c) both critical modes are rectangles. 

In each case, we explicitly find nondimensional numbers which determine the number, patterns and the stabilities of the bifurcated steady states. We also determine the basin of attraction of each of the stable steady states. 

In all the scenarios, we found that after the transition, only pure modes (rolls or rectangles) are stable and the mixed modes are unstable. Our result is conclusive when one of the critical modes is a roll type. When both critical modes are rectangles, we only have computational evidence.

When both critical modes are rolls, the stable steady states after the transition are rolls. When both critical modes are rectangles, computational evidence suggests that the stable steady states after the transition are rectangles. When one critical mode is a roll and the other one is a rectangle, the stable states after the transition can be either only rolls or both rolls and rectangles. 

The problem is usually studied in the infinitely extended horizontal domain setting which eliminates the effects of the boundaries in the horizontal directions. Our setting is a 3D rectangular domain with free-slip boundary conditions for the velocity, that is the fluid can not cross the boundaries but is allowed to slip. The thermal boundary conditions are adiabatically isolated side walls so that no heat is transferred through them and perfectly conducting top and bottom boundaries.

Technically, the analysis is carried out using the dynamical transition theory (Ma and Wang \cite{ptd}). One key ingredient is the reduction of the original system to the center manifold generated by the two unstable modes. The only modification that has been made is (following Sengul and Wang \cite{Sengul2011}), we  expand the center manifold using a basis which differs from the eigenfunctions of the original linear operator. This allows us to passby the difficulties associated with determining the eigenpairs in terms of the system parameters. We also make use of computer assistance, namely a Mathematica code, which carries out numerous integrations which are due to the interactions of the critical modes with the non-critical ones.

The paper is organized as follows: In Section 2, the governing equations and the functional setting of the problem is introduced. Section 3 deals with the linear theory. We present our main results in Section 4. The proof of these theorems are given in Section 5. In section 6, we present the physical conclusions derived from our theorems. Finally, Section 7 is the conclusion section. 

\section{Governing Equations and The Functional Setting}
With the Boussinesq approximation, the non-dimensional equations governing the motion and the states  of the Rayleigh-B{\'e}nard convection in a nondimensional  rectangular domain
$\Omega =\left( 0,L_{1}\right)
\times \left( 0,L_{2}\right) \times \left( 0,1\right) \subset \mathbb{R}^{3}$
are given as follows; see \cite{Chandrasekhar1981} among others: 
\begin{equation}  \label{main}
\begin{aligned} 
& \frac{\partial \mathbf{u}}{\partial t}+\left(
\mathbf{u}\cdot \nabla \right) \mathbf{u} = \text{Pr}\left( -\nabla p+\Delta
\mathbf{u}+R \theta {\bf k}\right) , \\ 
& \frac{\partial \theta }{\partial t}+\left(
\mathbf{u}\cdot \nabla \right) \theta = w+\Delta \theta , \\ & \nabla \cdot
\mathbf{u} =0, \\ & \mathbf{u}\left( 0\right) =\mathbf{u}_{0}\text{, \ \
}\theta \left( 0\right) =\theta _{0}. \end{aligned}
\end{equation}

The unknown functions are the velocity $\mathbf{u}=\left( u,v,w\right) $,  the temperature $\theta $,  and the pressure $p$. These unknowns represent a deviation 
from a motionless state basic steady state with  a constant positive vertical temperature
gradient.  In addition ${\bf k}$ stands for the unit vector in the $z$-direction.

The non-dimensional numbers in \eqref{main} are the Rayleigh number $R$ which is the control parameter and $\Pr$,  the Prandtl number.

The above system is supplemented with a set of boundary conditions. We use the free-slip boundary conditions for the velocity on all the boundaries. Thermally, the top and the bottom boundaries are assumed to be perfectly conducting and the horizontal boundaries are adiabatically isolated. Namely, the boundary conditions are as follows:
\begin{equation}  \label{bc}
\begin{aligned} 
& u =\frac{\partial v}{\partial x}=\frac{\partial w}{\partial x}=\frac{\partial \theta }{\partial x}=0 \,&&\text{at}\,x=0,L_{1}, \\ 
& \frac{\partial u}{\partial y}=v=\frac{\partial w}{\partial y}=\frac{\partial \theta }{\partial y}=0 \,&&\text{at}\,y=0,L_{2}, \\ &\frac{\partial u}{\partial z}=\frac{\partial v}{\partial z}=w=\theta =0  \,&&\text{at}\,z=0,1.
\end{aligned}
\end{equation}

For the functional setting, we define the relevant function spaces:
\begin{equation}\label{func spaces}
\begin{aligned} 
& H =\left\{ (\mathbf{u},\theta)\in L^{2}\left(
\Omega,\mathbb{R}^4 \right) :\nabla\cdot\mathbf{u}=0,\mathbf{u}\cdot n\mid_{\partial \Omega}=0\right\},\\
& H_1 =\left\{ (\mathbf{u},\theta)\in H^{2}\left(
\Omega, \mathbb{R}^4 \right) :\nabla\cdot\mathbf{u}=0,\mathbf{u}\cdot n\mid_{\partial \Omega}=0,\theta\mid_{z=0,1}=0\right\}.
\end{aligned}
\end{equation}

For $\phi=({\bf u},\theta)$, let $G:H_1\rightarrow H$ and $L_R:H_1\rightarrow H$ be defined by:
\begin{equation}\label{operators}
\begin{aligned}
& L_R\phi=(\Pr\mathcal{P}(\Delta {\bf u}+R\theta {\bf k}),w+\Delta\theta),\\
& G(\phi)=-(\mathcal{P}({\bf u}\cdot\nabla){\bf u},({\bf u}\cdot\nabla)\theta),
\end{aligned}
\end{equation}
with $\mathcal{P}$ denoting the Leray projection onto the divergence-free vectors.
The equations \eqref{main} and \eqref{bc} can be put into the following functional form:
\begin{equation} \label{functional}
\frac{d\phi}{dt}=L_R\phi+G(\phi),\qquad \phi(0)=\phi_0.
\end{equation}

The results concerning existence and uniqueness of \eqref{functional} are classical and we refer the interested readers to Foias, Manley and Temam \cite{Foias1987} for details. In particular, we can define a semigroup:
\[ S(t):\phi_0\rightarrow \phi(t). \]

Finally we define the following trilinear forms which will be used in the proof of the main  theorems:
\begin{equation} \label{bilinear}
\begin{aligned}
& G(\phi_{1},\phi_{2},\phi_{3})=-\int_{\Omega}({\bf u_{1}}\cdot\nabla){\bf u_{2}}\cdot{\bf u_{3}}-\int_{\Omega}({\bf u_{1}}\cdot\nabla)\theta_{2}\cdot \theta_{3},\\
& G_s(\phi_1,\phi_2,\phi_3)= G(\phi_1,\phi_2,\phi_3)+ G(\phi_2,\phi_1,\phi_3).
 \end{aligned}
\end{equation}

\section{Linear Theory}
We  recall in this section the well-known linear theory of the problem. 
\subsection{Linear Eigenvalue Problem}
We first study the eigenvalue problem:
\begin{equation}\label{EV}
\begin{aligned} 
& \Pr(\Delta {\bf u}+R\,\theta\, {\bf k}-\nabla p)=\beta {\bf u},\\
& w+\Delta\theta=\beta \theta,\\
& \text{div} {\bf u}=0,
\end{aligned}  
\end{equation}
with the boundary conditions \eqref{bc}.
Thanks to the boundary conditions, we can represent the solutions $\phi_S=({\bf u}_S,\theta_S)$, ${\bf u}_S=(u_S,v_S,w_S)$ by the separation of variables:
\begin{equation} \label{sepofvar}
\begin{aligned} 
& u_S =U_S \sin (L _{1}^{-1}s_x\pi x)\cos (L _{2}^{-1}s_y\pi y) \cos  (s_z\pi z), \\ 
& v_S =V_S \cos (L _{1}^{-1}s_x\pi x)\sin (L _{2}^{-1}s_y\pi y) \cos (s_z\pi z), \\ 
& w_S =W_S \cos (L _{1}^{-1}s_x\pi x)\cos (L _{2}^{-1}s_y\pi y) \sin (s_z\pi z), \\ 
& \theta_S =\Theta_S  \cos (L _{1}^{-1}s_x\pi x)\cos (L _{2}^{-1}s_y\pi y) \sin  (s_z\pi z),
\end{aligned}  
\end{equation} 
for $S=(s_x,s_y,s_z)$ where $s_x\geq0$, $s_y\geq0$, $s_z\geq0$. It is easy to see that only eigenvalues $\beta_S$, $S\in\mathcal{Z}$ can become positive where
\[ \mathcal{Z}=\{(s_x,s_y,s_z)\mid s_x\geq0, s_y\geq0,s_z\geq0,\, (s_x,s_y)\neq(0,0) \text{ and }s_z\neq0\}.\]
For $S=(s_x,s_y,s_z)\in\mathcal{Z}$, the amplitudes of the horizontal velocity field can be found as:
\begin{equation*}
U_S =-\frac{s_x\pi}{L _{1}} \frac{s_z\pi}{\alpha_S^2} W_S , \qquad  V_S=-\frac{s_y \pi}{L _{2}} \frac{s_z\pi}{\alpha_S^2}W_S.
\end{equation*}
We define $\alpha_S$, the horizontal wave number and $\gamma_S$, the full wave number by:
\begin{equation}\label{wavenumbers}
\alpha_S=\sqrt{\frac{s_x^2\pi^2}{L_1^2}+\frac{s_y^2\pi^2}{L_2^2}},\qquad \gamma_S=\sqrt{\frac{s_x^2\pi^2}{L_1^2}+\frac{s_y^2\pi^2}{L_2^2}+s_z^2\pi^2}.
\end{equation}
Taking the divergence of the first equation in \eqref{EV} we find:
\[
\Delta p=R\frac{\partial\theta}{\partial z}. \]
Now taking the Laplacian of the first equation, replacing $\Delta p$ by the above relation and using \eqref{sepofvar} we obtain:
\begin{equation}\label{precharequ}
\begin{aligned}
& \gamma_S^2(\Pr\,\gamma_S^2+\beta) W_S-R\,\Pr\,\alpha_S^2\Theta_S=0,\\
& W_S-(\gamma_S^2+\beta)\Theta_S=0.
\end{aligned}
\end{equation}
For each $S\in\mathcal{Z}$, the above equations have two solutions $\beta_S^1>\beta_S^2$ which satisfy the following  equation:
\begin{equation}\label{charequ}
 \gamma_S^2(\gamma_S^2+\beta)(\Pr\,\gamma_S^2+\beta)-R\,\Pr\,\alpha_S^2=0.
\end{equation}
We find that amplitudes of the normalized critical eigenvectors as:
\begin{equation}
W_S=\beta_S^1(R)+\gamma_S^2,\qquad \Theta_S=1.
\end{equation}
Now solving \eqref{charequ} for $R$ at $\beta=0$, the critical Rayleigh number can be defined as:
\begin{equation}\label{Rr}
R_c=\underset{S\in \mathcal{Z}}{\min} R_S, \qquad R_S:=\frac{\gamma_S^6}{\alpha_S^2}.
\end{equation}
From \eqref{Rr}, one sees that for a minimizer $S=(s_x,s_y,s_z)$ of $R_S$, the vertical index $s_z$ is 1. We will denote the set of critical indices $S$ minimizing \eqref{Rr} by $\mathcal{C}$:
\[
\mathcal{C}=\{S=(s_x,s_y,1)\in\mathcal{Z}\mid R_S\leq R_{S^{\prime}}, \, \forall S^{\prime}\in\mathcal{Z}\}
\]
For small length scale region, the map in Figure \ref{wavemapcolored} shows the horizontal critical wave indices that are picked by the selection mechanism \eqref{Rr}.

\begin{figure}
\includegraphics[scale=.5]{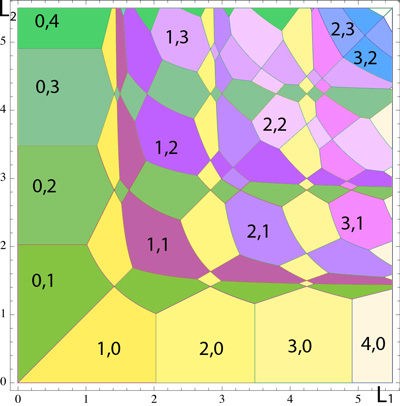}
\caption{The selection of critical horizontal wave indices in the $L_1$--$L_2$ plane. The same coloring indicates equal wave indices.\label{wavemapcolored}}
\end{figure}

It is well known that we have the following PES condition:
\begin{align}
& \beta _S^1(R) =\left\{ 
\begin{aligned}
& <0, && \lambda <R_c, \\ 
& =0, && \lambda =R_c,\\ 
& >0, && \lambda >R_c,
\end{aligned}\right. 
          &&  \forall S\in\mathcal{C},\label{PES1} \\
& \Re \beta(R_c) <0, && \forall \beta\notin\{\beta_S^1\mid S\in \mathcal{C}\}.
\label{PES2}
\end{align}
By \eqref{charequ}, corresponding to $S=(s_x,s_y,s_z)\in\mathcal{Z}$, there are two eigenvalues $\beta_S^i$, and two corresponding eigenfunctions $\phi_S^i$, $i=1,2$. If a critical mode has wave index $I$ then the corresponding eigenfunction is $\phi_I^1$ which we will simply denote by $\phi_I$.

Depending on the horizontal wave indices, there are two types of critical modes corresponding to two different patterns. If the wave index $I=(i_x,i_y,1)$ of a critical mode is such that one of the horizontal wave indices $i_x$, $i_y$ is zero, the corresponding eigenfunction has a roll pattern. When both horizontal indices are non-zero, the corresponding eigenfunction has a rectangular pattern. Figure \ref{RollRecPat} shows a sketch of these patterns.

\begin{figure}
\centering
\subfigure[A roll pattern with wave index $J=(0,j_y,1)$.]{
\includegraphics[scale=0.3]{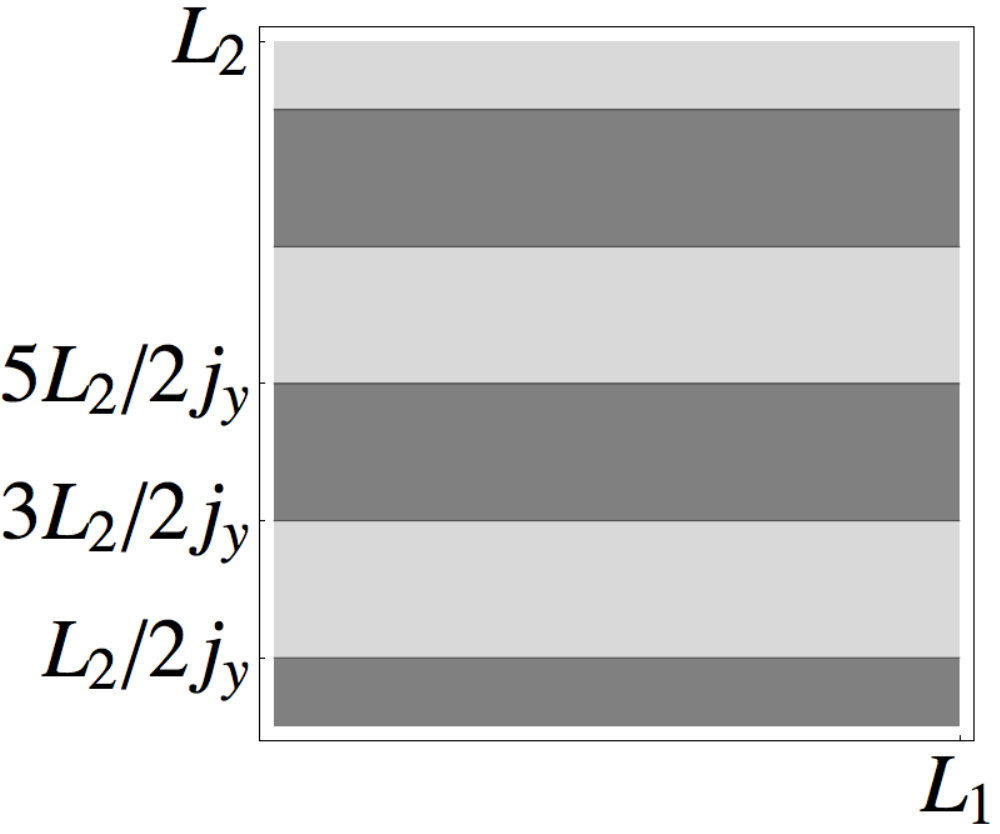}
}
\qquad
\subfigure[A rectangle pattern with wave index $I=(i_x,i_y,1)$. ]{
\includegraphics[scale=0.3]{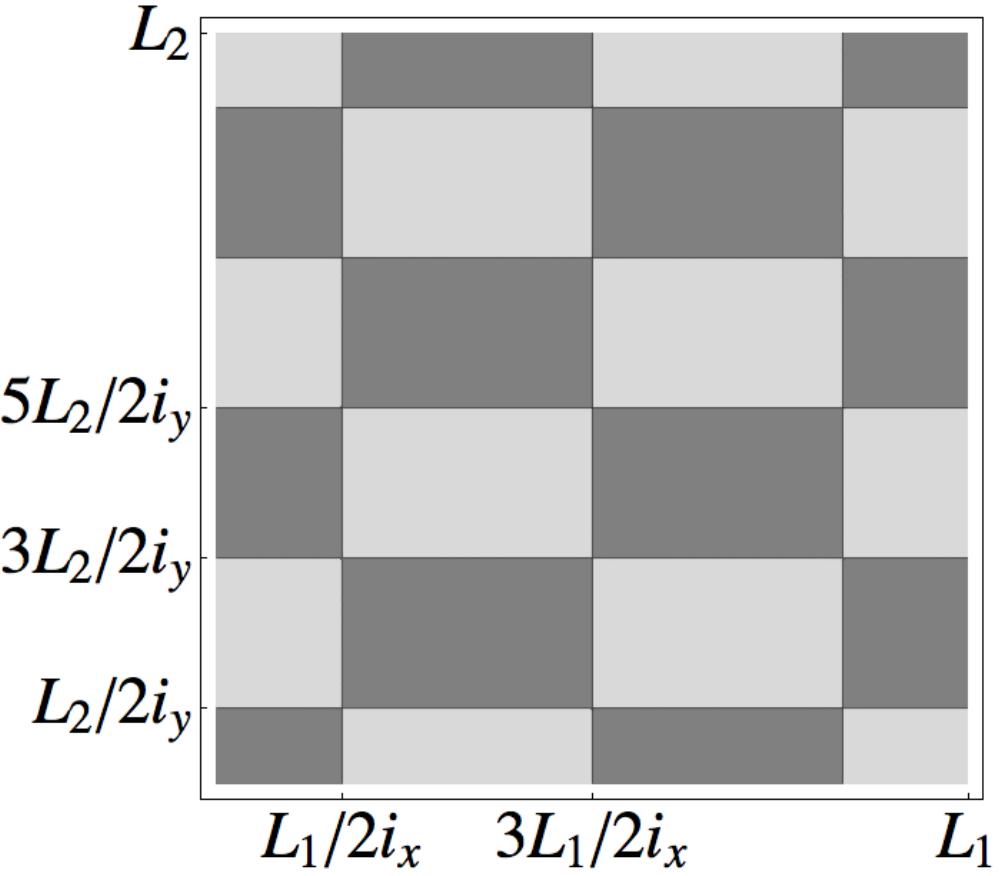}
}
\caption{Regions of positive and negative vertical velocity of a rectangular and a roll mode at the mid-plane $z=1/2$. \label{RollRecPat}}
\end{figure}
\subsection{Estimation of the critical wave number}
As it will be shown, the dynamic transitions depend on the critical wavenumber $\alpha$. In the case of infinite horizontal domains, the critical wave number is found to be $\alpha=\pi/\sqrt{2}\approx 2.22$ corresponding to a critical Rayleigh number $R_c=27\pi^4/4\approx658$. For rectangular domains, the wave number is not constant and is a function of the length scales. Following estimates will be important in the physical remarks section:

\begin{lemma}\label{wave number estimate lemma}
Let $\alpha$ be the critical wave number minimizing \eqref{Rr}. Then
\[
\begin{aligned}
& \alpha\geq \frac{\pi}{2^{1/3}(2^{2/3}+1)^{1/2}}\approx1.55, \quad \text{ for all } L_1,L_2,\\
& \alpha < \frac{2^{2/3}\pi}{\sqrt{1+2^{2/3}}}\approx3.10, \quad \text{ if }L_1>2^{1/3}\sqrt{1+2^{2/3}}\approx2.03,\\
& \alpha\rightarrow\frac{\pi}{\sqrt{2}}, \quad L_1\rightarrow\infty.
\end{aligned}
\] 
\end{lemma}
\begin{proof}
To estimate the dependence of the wave number $\alpha$ on the length scales $L_1$, $L_2$, we define:
$$L(m)=((m+1)m)^{1/3}((m+1)^{2/3}+m^{2/3})^{1/2}, \quad m\in\mathbb{Z},m\geq 0.$$
The sequence $L(m)$ gives those length scales of $L_1$ for which the wave index changes assuming $L_2$ is sufficiently small.
As shown in Sengul and Wang \cite{Sengul2011}, when $L(m-1)<L_1<L(m)$ for some $m\geq1$, we have the following bound on the critical wave number:
\begin{equation}\label{wavenumberest}
\frac{m \pi}{L(m)}< \alpha < \frac{m \pi}{L(m-1)}.
\end{equation}
In particular,
\[
\alpha\geq \frac{\pi}{L(1)}=\frac{\pi}{2^{1/3}(2^{2/3}+1)^{1/2}},
\]
and
\[
\alpha \leq \frac{2\pi}{L(1)}=\frac{2^{2/3}\pi}{\sqrt{1+2^{2/3}}}, \quad \text{if }L_1>L(1).
\]
Finally noticing that,
\[\frac{m\pi}{L(m)}\rightarrow \frac{\pi}{\sqrt{2}}, \quad \frac{(m+1)\pi}{L(m)}\rightarrow \frac{\pi}{\sqrt{2}}, \quad \text{as } m\rightarrow \infty,\]
we find that 
\[\alpha\rightarrow\frac{\pi}{\sqrt{2}}, \quad L_1\rightarrow\infty. \]
\end{proof}
The bounds on the critical wave number as a function of the length scale $L_1$ which is obtained from \eqref{wavenumberest} is shown in Figure \ref{wavenumberestimate}.

\begin{figure}
\includegraphics[scale=.8]{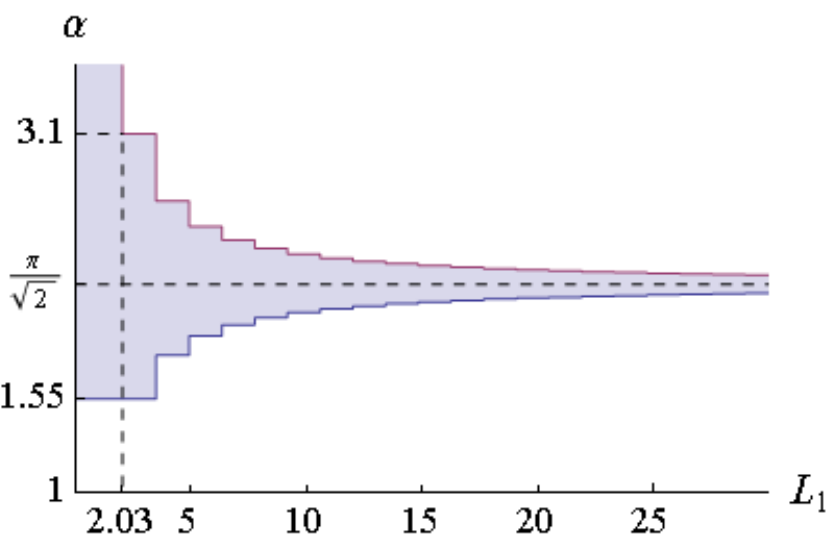}
\caption{The shaded region shows the possible values for the critical wave number $\alpha$ for a given $L_1$.\label{wavenumberestimate}}
\end{figure}

\section{Dynamic Transitions and Pattern Selection}
We study the case where two eigenvalues with indices $I=(i_x,i_y,1)$ and $J=(j_x,j_y,1)$ are the first critical eigenvalues. This means that $\alpha_I$ and $\alpha_J$ minimize \eqref{Rr}, thus the PES conditions \eqref{PES1}, \eqref{PES2} are satisfied with $\mathcal{C}=\{I, J\}$. The crucial assumption is that the corresponding wave numbers are equal, i.e. 
\[\alpha=\alpha_I=\alpha_J.\] 
Since $I\neq J$, without loss of generality, we can assume that $i_x>j_x$ which ensures that $i_y<j_y$. By \eqref{wavenumbers}, we must have the following linear relation between the length scales:
\[
L_1=\sqrt{\frac{i_x^2-j_x^2}{j_y^2-i_y^2}}L_2. 
\]
Thus two critical eigenmodes are possible only when $L_1$ and $L_2$ lie on a line emanating from the origin in Figure~\ref{wavemap1}.
There are three possible cases depending on the structure of the critical eigenmodes, which are completely described by our main theorems. \\
 (a) a rectangle and a roll mode respectively (described by Theorem \ref{roll vs rec thm}),\\ 
 (b) both  roll modes (described by Theorem \ref{roll vs roll thm}),\\
  (c) both rectangle modes (described by Theorem \ref{rec vs rec thm}).\\
 These possible cases are illustrated by Figure \ref{wavemap1} in the small length scale regime.
\begin{figure}
\includegraphics[scale=.5]{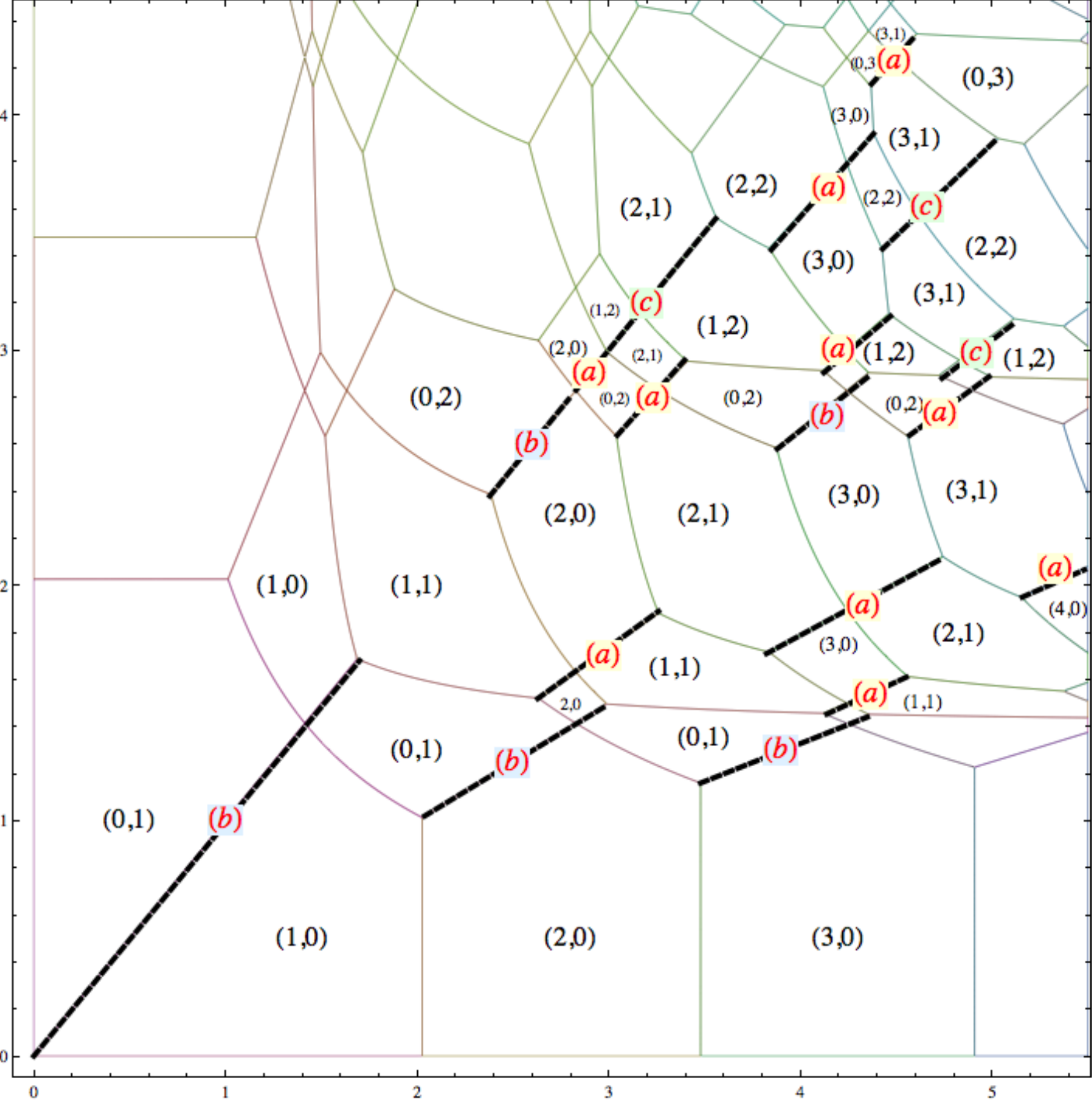}
\caption{The first two critical modes with (a) a rectangle and a roll pattern, (b) both roll patterns, (c) both rectangle patterns.\label{wavemap1}}
\end{figure}

Before presenting our results, we first summarize some of the known results which applies for the above setting; see Ma and Wang \cite{Ma2004a, Ma2007}:

\begin{itemize}
\item[i)] As the Rayleigh number $R$ crosses $R_c$, the system undergoes a Type-I (continuous) transition.
\item[ii)] There is an attractor $\Sigma_R$ bifurcated on $R>R_c$ such that for any $\phi_0\in H\setminus\Gamma$,
\[ dist(\phi_0,\Sigma_R)\rightarrow 0, \quad \text{as } t\rightarrow \infty,\] 
where $\Gamma$ is the stable manifold of $\phi=0$ with $codim=2$.
\item[iii)] $\Sigma_R$ is homeomorphic to $S^1$ and comprises steady states and the heteroclinic orbits connecting these steady states.
\item[iv)] There are four or eight bifurcated steady states. Half of the bifurcated steady states are minimal attractors and the rest are saddle points.
\end{itemize}
The dynamic transitions depend on the following \emph{positive} parameter which in turn is a function of the parameters $\Pr$, $L_1$ and $L_2$:
\begin{equation}\label{kappa}
\kappa_{S} =
\left\{\begin{array}{cc}
8\Pr^2 \alpha^2, & S=(0,0,2), \\
 \frac{\pi^2(4\alpha^2-\alpha_S^2)^2}{(R_S-R_c)\alpha^4}\left(\frac{\gamma_S^2\gamma^8}{R_c\alpha^2}+2\Pr\,\gamma^4+\Pr^2\frac{R_S\alpha^2}{\gamma_S^2}\right), & S\neq (0,0,2).
\end{array} \right.
\end{equation}
Here, $\alpha$ is the critical wave number, $\gamma^2=\alpha^2+\pi^2$ and $R_c=\gamma^6/\alpha^2$ is the critical Rayleigh number. Moreover, for $S=(s_x,s_y,s_z)$, $(s_x,s_y)\neq(0,0)$:
\[ \alpha_S^2= \frac{s_x^2 \pi^2}{L_1^2}+\frac{s_y^2\pi^2}{L_2^2},  \quad \gamma_S^2=\alpha_S^2+s_z^2 \pi^2, \quad R_S=\frac{\gamma_S^6}{\alpha_S^2}.\]
Also we let
\[g=\frac{\Pr \alpha^2}{(\Pr+1)\gamma^4}.\]
\subsection{One of the critical modes is a roll, the other is a rectangle.}
We first consider the case where an eigenmode with a roll structure and an eigenmode with a rectangle structure are the first critical eigenmodes. 
 \begin{figure}
\centering
\subfigure[When $a<c$. ]{
\includegraphics[scale=0.47]{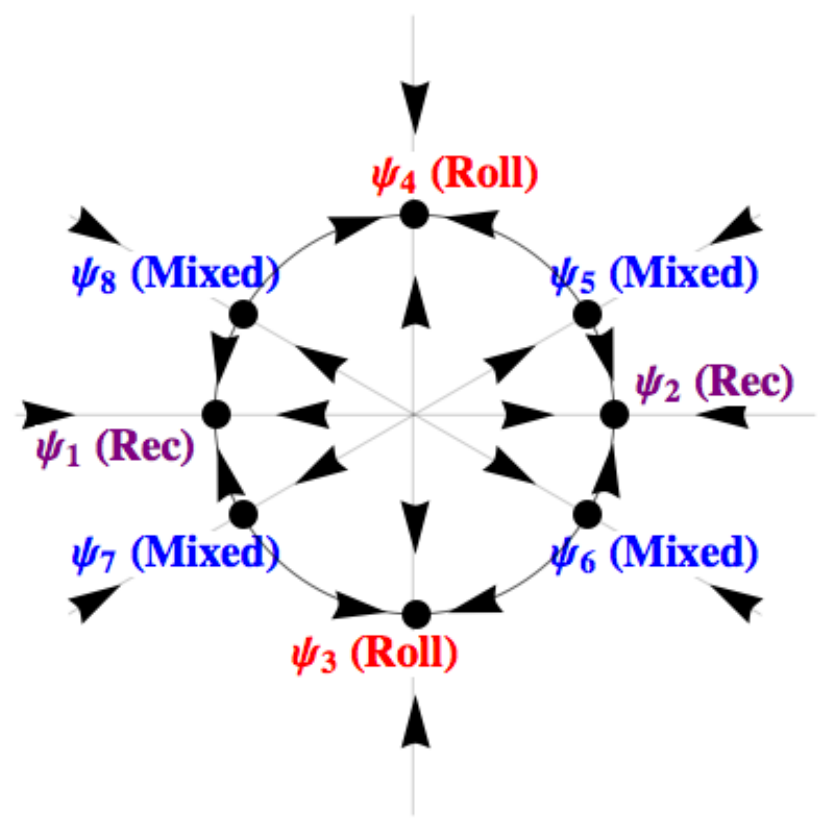}
\label{ac}
}
\subfigure[When $c<a$.]{
\includegraphics[scale=0.47]{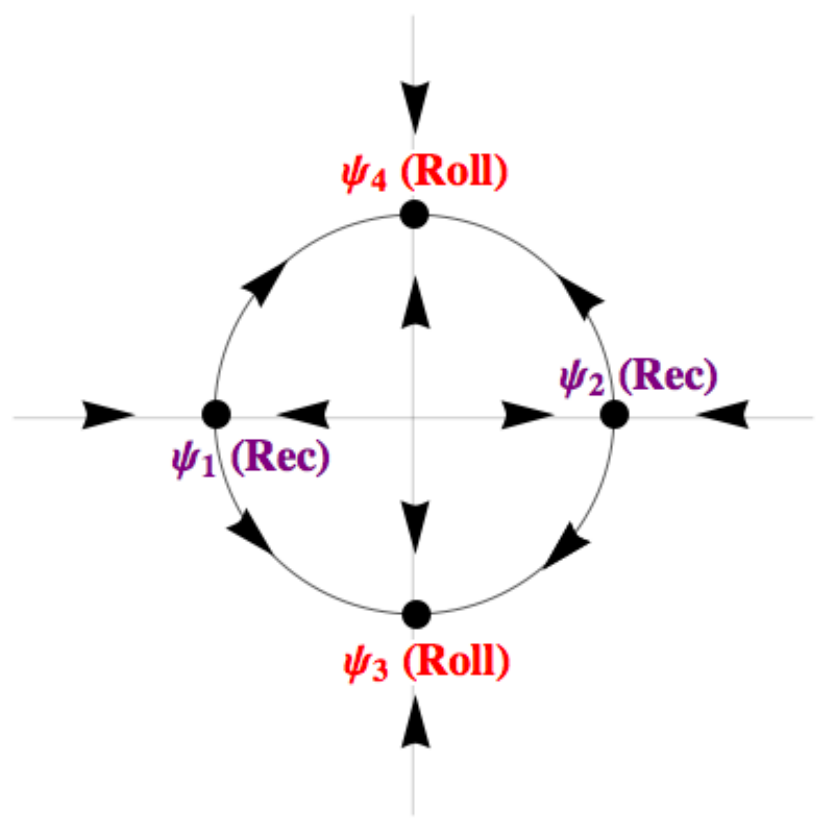}
\label{ca}
}
\caption{The structure of the attractor after the transition $R>R_c$ for different parameter regions when the first two critical modes are a roll and a rectangle. \label{regions}}
\end{figure}

\begin{theorem} \label{roll vs rec thm}
Assume that $I=(i_x,i_y,1)$ and $J=(0,j_y,1)$, ($i_x\geq1$, $j_y>i_y\geq1$) are the first critical indices with identical wave numbers, $\alpha_I=\alpha_J$. Consider the following numbers:
\begin{equation}\label{definition abc}
\begin{aligned}
& a= \kappa_{0,0,2}+\kappa_{2i_x,0,2}+\kappa_{0,2i_y,2},\\
& b= \kappa_{0,0,2},\\
& c = \kappa_{0,0,2}+2\kappa_{i_x,i_y+j_y,2}+2\kappa_{i_x,-i_y+j_y,2}.\\
\end{aligned}
\end{equation}
For $R>R_c$, let us define the following steady state solutions:
\[
\psi_i=g\sqrt{R-R_c}(X_i \phi_I+Y_i \phi_J)+o((R-R_c)^{1/2}), \quad i=1,\dots,8,
\]
where 
\begin{align*}
& X_i=(-1)^{i}a^{-1/2},  && Y_i=0, &&& i=1,2,  &&&&\text{(rectangle pattern)}\\
& X_i=0, && Y_i=(-1)^i(2b)^{-1/2}, &&& i=3,4, &&&&\text{(roll pattern)} \\
& X_i=\sqrt{\frac{c-b}{c^2-a b}}, && Y_i=(-1)^i\sqrt{\frac{c-a}{2(c^2-ab)}}, &&&i=5,6,  &&&&\text{(mixed pattern)}\\
& X_i=-\sqrt{\frac{c-b}{c^2-a b}}, && Y_i=(-1)^i\sqrt{\frac{c-a}{2(c^2-ab)}}, &&&i=7,8,  &&&&\text{(mixed pattern)}
\end{align*}
There are two possible transition scenarios:
\begin{itemize}
\item[i)] If $a<c$ then the topological structure of the system after the transition is as in Figure \ref{ac}. In particular:
\begin{itemize}
\item[1)] $\Sigma_R$ contains eight steady states $\psi_i$, $i=1,\dots,8$. 
\item[2)] $\psi_1$, $\psi_2$ (rectangles) and $\psi_3$, $\psi_4$ (rolls) are minimal attractors of $\Sigma_R$, $\psi_5$, $\psi_6$, $\psi_7$, $\psi_8$ (mixed) are unstable. 
\item[3)] There is a neighborhood $\mathcal{U}\setminus\Gamma$ of $0$ where $\Gamma$ is the stable manifold of $0$ such that $\bar{\mathcal{U}}=\cup_{i=1}^4 \bar{\mathcal{U}_i}$ with $\mathcal{U}_i$ pairwise disjoint and $\mathcal{U}_i$ is the basin of attraction of $\psi_i$, $i=1,\dots,4$.
\item[4)]  The projection of $\mathcal{U}_i$ onto the space spanned by $\phi_I$, $\phi_J$ is approximately a sectorial region given by:\[ \mathcal{U}_i\cap\{X\phi_I+Y\phi_J\mid \omega_{i,1}<arg(X,Y)<\omega_{i,2}\}, \quad i=1,\dots,4,\]
\begin{align*}
& \omega_{1,1}=\pi-\omega, &&\omega_{1,2}=\pi+\omega, &&\omega_{2,1}=-\omega, &&\omega_{2,2}=\omega, \\
&  \omega_{3,1}=\pi+\omega, &&\omega_{3,2}=2\pi-\omega, &&\omega_{4,1}=\omega,  &&\omega_{4,2}=\pi-\omega,
\end{align*}
where 
\begin{equation} \label{omega}
\omega=\arctan{\sqrt{\frac{c-a}{2(c-b)}}}.
\end{equation}
\end{itemize} 
\item[ii)] If $c<a$ then the topological structure of the system after the transition is as in Figure \ref{ca}. In particular:
\begin{itemize}
\item[a)] $\Sigma_R$ contains four steady states $\psi_i$, $i=1,\dots,4$.
\item[b)] $\psi_3$, $\psi_4$ (rolls) are minimal attractors of $\Sigma_R$ whereas the $\psi_1$, $\psi_2$ (rectangles) are unstable steady states. 
\item[c)] There is a neighborhood $\mathcal{U}\setminus\Gamma$ of $0$ where $\Gamma$ is the stable manifold of $0$ such that $\bar{\mathcal{U}}=\cup_{i=3}^4 \bar{\mathcal{U}_i}$ with $\mathcal{U}_i$ pairwise disjoint and $\mathcal{U}_i$ is the basin of attraction of $\psi_i$, $i=3,4$. 
\item[d)] The projection of $\mathcal{U}_i$ onto the space spanned by $\phi_I$, $\phi_J$ is a sectorial region given by:
\[ \mathcal{U}_i\cap\{X\phi_I+Y\phi_J\mid \omega_{i,1}<arg(X,Y)<\omega_{i,2}\}, \quad i=3,4,\]
\[
 \omega_{3,1}=\pi,\,  \omega_{3,2}=2\pi, \quad \omega_{4,1}=0,\, \omega_{4,2}=\pi. 
\]
\end{itemize}

\end{itemize}
\end{theorem}

\begin{remark}
In the special case $j_y=2i_y$, the mixed solution corresponds to a regular hexagonal pattern. In this case we find $a<c$, hence the first scenario in Theorem~\ref{roll vs rec thm} is valid; see Remark~\ref{Rem}.
\end{remark}

\subsection{The first two critical modes are both rolls.} In this section we consider two critical modes both having a roll structure. Under the assumption that the wave numbers are equal, one of the rolls has to be aligned in the x-direction and the other one aligned in the y-direction.

\begin{figure}
\includegraphics[scale=0.55]{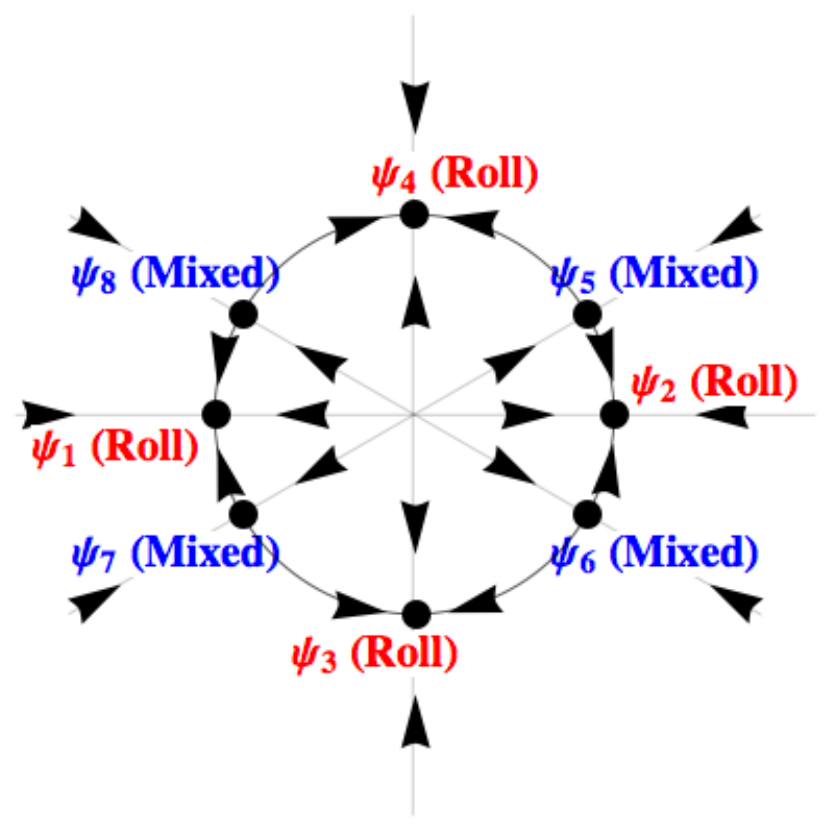}
\label{bd region}
\caption{The structure of the attractor after the transition $R>R_c$ when the first two critical modes are both roll type. \label{roll vs roll regions}}
\end{figure}

\begin{theorem} \label{roll vs roll thm}
Assume that $I=(i_x,0,1)$ and $J=(0,j_y,1)$ ($i_x\geq1$, $j_y\geq1$) are the first critical indices with identical wave numbers, $\alpha_I=\alpha_J$. Consider the following numbers:
\begin{equation} \label{definition bd}
\begin{aligned}
& b= 2\kappa_{0,0,2},\\
& d = 2\kappa_{0,0,2}+8\kappa_{i_x,j_y,2}.\\
\end{aligned}
\end{equation}
For $R>R_c$, we define:
\[
\psi_i=g\sqrt{R-R_c}(X_i \phi_I+Y_i \phi_J)+o((R-R_c)^{1/2}),\quad i=1,\dots,8,
\]
where
\begin{align*}
& X_i=(-1)^{i} b^{-1/2},  && Y_i=0, &&& i=1,2,  &&&&\text{(roll pattern)}\\
& X_i=0, && Y_i=(-1)^i b^{-1/2}, &&& i=3,4, &&&&\text{(roll pattern)} \\
& X_i=(b+d)^{-1/2}, && Y_i=(-1)^i X_i, &&&i=5,6,  &&&&\text{(mixed pattern)}\\
& X_i=-(b+d)^{-1/2}, && Y_i=(-1)^i X_i,  &&&i=7,8,  &&&&\text{(mixed pattern)}
\end{align*}
Then the topological structure of the system after the transition is as in Figure \ref{bd region}. In particular:
\begin{itemize}
\item[1)] $\Sigma_R$ contains eight steady states $\psi_i$, $i=1,\dots,8$. 
\item[2)] $\psi_1$, $\psi_2$, $\psi_3$, $\psi_4$ (rolls) are minimal attractors of $\Sigma_R$, $\psi_5$, $\psi_6$, $\psi_7$, $\psi_8$ (mixed) are unstable. 
\item[3)] There is a neighborhood $\mathcal{U}\setminus\Gamma$ of $0$ where $\Gamma$ is the stable manifold of $0$ such that $\bar{\mathcal{U}}=\cup_{i=1}^4 \bar{\mathcal{U}_i}$ with $\mathcal{U}_i$ pairwise disjoint and $\mathcal{U}_i$ is the basin of attraction of $\psi_i$, $i=1,\dots,4$.
\item[4)]  The projection of $\mathcal{U}_i$ onto the space spanned by $\phi_I$, $\phi_J$ is approximately a sectorial region given by:\[ \mathcal{U}_i\cap\{X\phi_I+Y\phi_J\mid \omega_{i,1}<arg(X,Y)<\omega_{i,2}\}, \quad i=1,\dots,4,\]
\begin{align*}
& \omega_{1,1}=3\pi/4, &&\omega_{1,2}=5\pi/4, &&\omega_{2,1}=-\pi/4, &&\omega_{2,2}=\pi/4, \\
&  \omega_{3,1}=5\pi/4, &&\omega_{3,2}=7\pi/4, &&\omega_{4,1}=\pi/4,  &&\omega_{4,2}=3\pi/4.
\end{align*}
\end{itemize} 
 \end{theorem}

\subsection{The first two critical modes are both rectangles.} In this section we consider two critical modes both having a rectangular pattern with equal wave numbers, $\alpha_I=\alpha_J$.

\begin{figure}
\centering
\subfigure[When $a<e$ and $f<e$.]{
\includegraphics[scale=0.55]{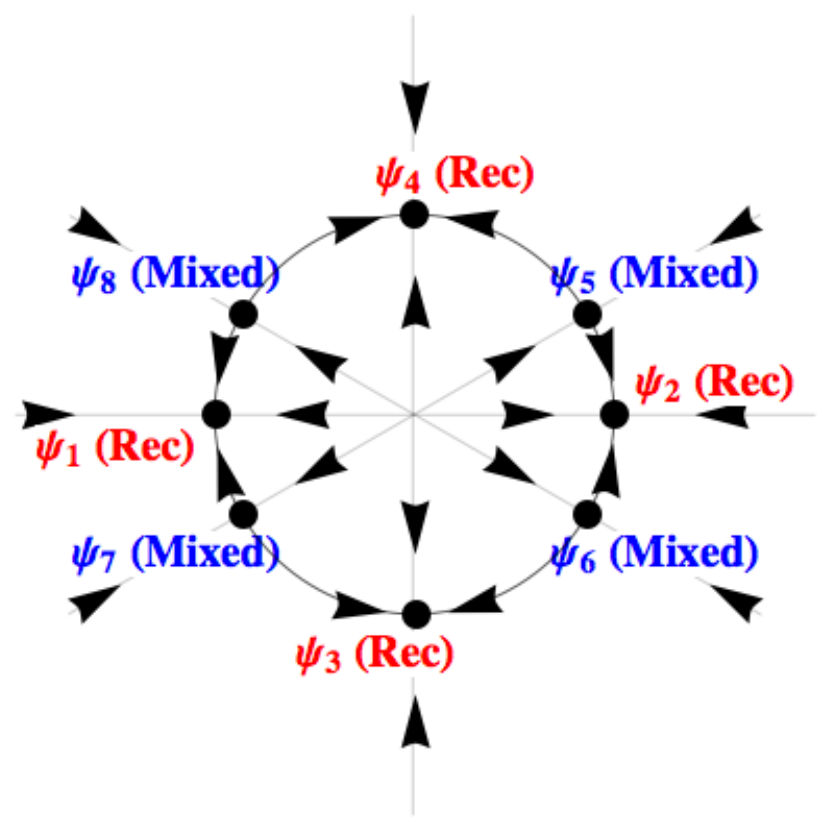}
\label{afe region}
}
\subfigure[When $e<a$ and $e<f$.]{
\includegraphics[scale=0.55]{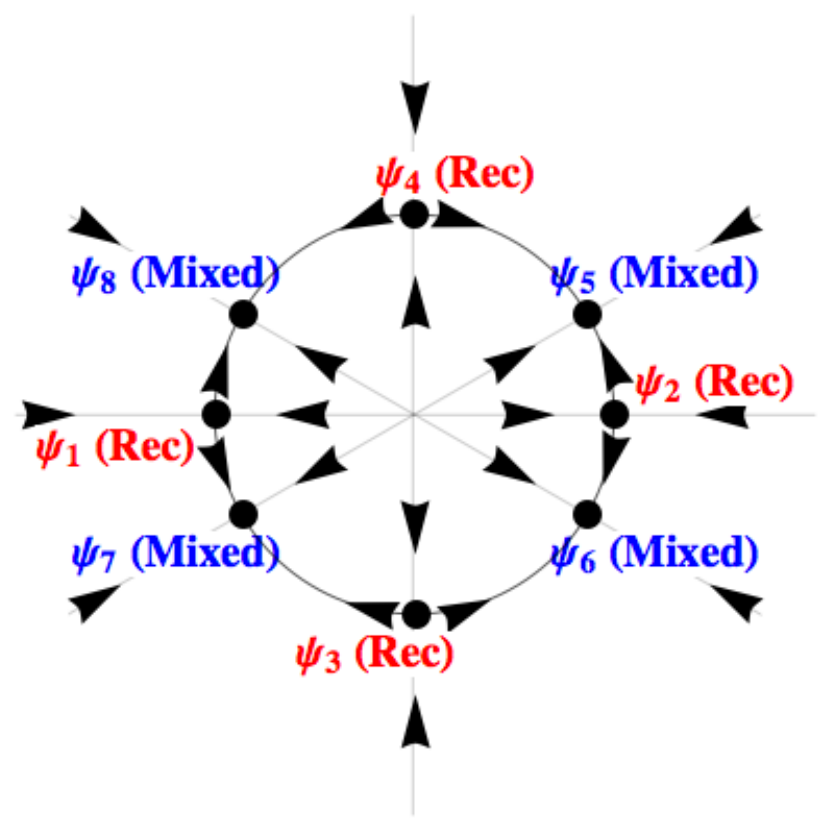}
\label{eaf region}
}
\subfigure[When $f<e<a$.]{
\includegraphics[scale=0.55]{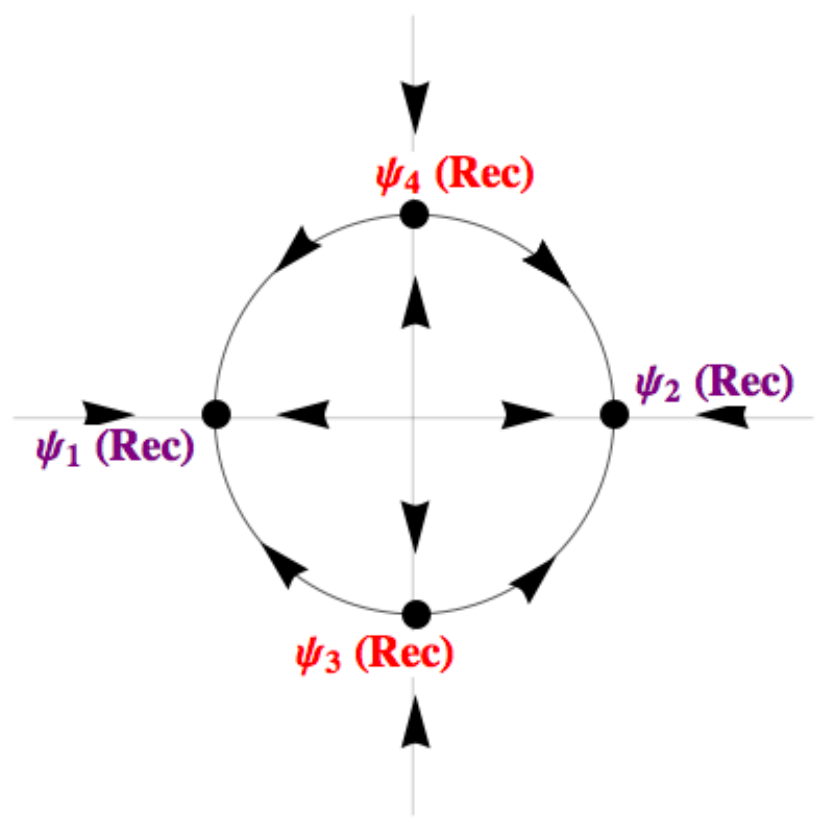}
\label{fea region}
}
\subfigure[When $a<e<f$.]{
\includegraphics[scale=0.55]{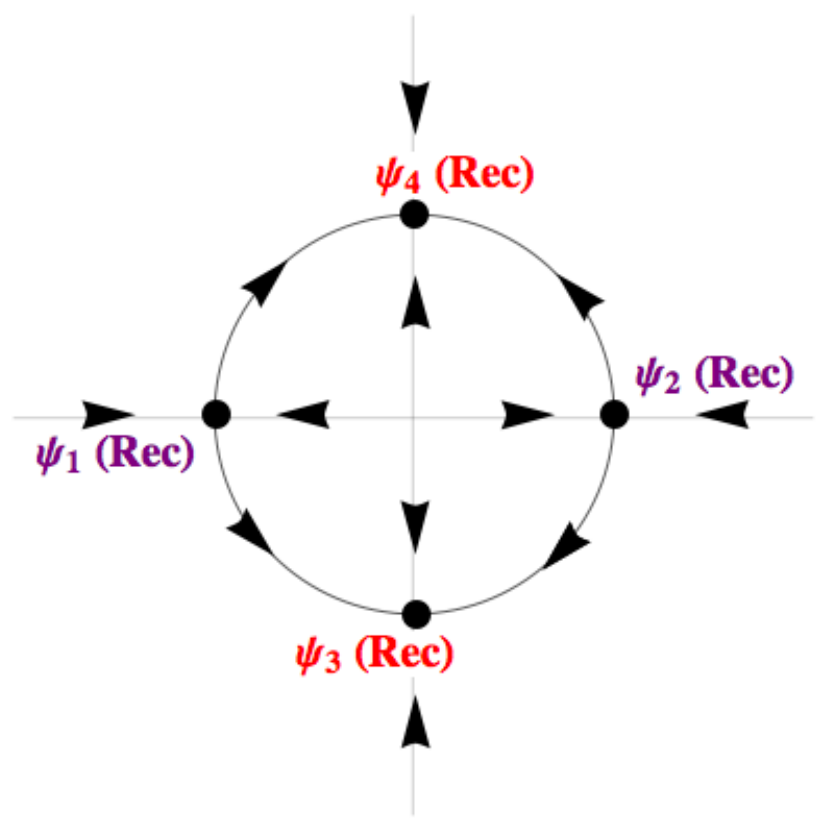}
\label{aef region}
}
\caption{The structure of the attractor after the transition $R>R_c$ for different parameter regions when the first two critical modes are both rectangle type. \label{rec vs rec regions}}
\end{figure}

\begin{theorem} \label{rec vs rec thm}
Assume that $I=(i_x,i_y,1)$, $J=(j_x,j_y,1)$ ($i_x\neq0$, $i_y\neq0$, $j_x\neq0$, $j_y\neq0$, $i_x\neq j_x$, $i_y\neq j_y$) are the first critical indices with identical wave numbers, $\alpha_I=\alpha_J$. Consider the following numbers:
\begin{equation} \label{definition aef}
\begin{aligned}
& a= \kappa_{0,0,2}+\kappa_{2i_x,0,2}+\kappa_{0,2i_y,2},\\
& e= \kappa_{0,0,2}+\kappa_{i_x+j_x,i_y+j_y,2}+\kappa_{i_x-j_x,i_y+j_y,2}+\kappa_{i_x+j_x,-i_y+j_y,2}+\kappa_{i_x-j_x,-i_y+j_y,2},\\
& f= \kappa_{0,0,2}+\kappa_{2j_x,0,2}+\kappa_{0,2j_y,2}.
\end{aligned}
\end{equation}
For $R>R_c$, let us define the following steady state solutions:
\begin{equation*}
\psi_i=g\sqrt{R-R_c}(X_i \phi_I+Y_i \phi_J)+o((R-R_c)^{1/2}), \quad i=1,\dots,8,
\end{equation*}
where
\begin{align*}
& X_i=(-1)^{i}a^{-1/2},  && Y_i=0, &&& i=1,2,  &&&&\text{(rectangle pattern)}\\
& X_i=0, && Y_i=(-1)^i f^{-1/2}, &&& i=3,4, &&&&\text{(rectangle pattern)} \\
& X_i=\sqrt{\frac{e-f}{e^2-a f}}, && Y_i=(-1)^i\sqrt{\frac{e-a}{2(e^2-af)}}, &&&i=5,6,  &&&&\text{(mixed pattern)}\\
& X_i=-\sqrt{\frac{e-f}{e^2-a f}}, && Y_i=(-1)^i\sqrt{\frac{e-a}{2(e^2-af)}}, &&&i=7,8,  &&&&\text{(mixed pattern)}
\end{align*}
There are four possible transition scenarios:
\begin{itemize}
\item[i)] If $a<e$ and $f<e$ then the topological structure of the system after the transition is as in Figure \ref{afe region}. In particular:
\begin{itemize}
\item[1)] $\Sigma_R$ contains eight steady states $\psi_i$, $i=1,\dots,8$. 
\item[2)] $\psi_1$, $\psi_2$, $\psi_3$, $\psi_4$ (rectangles) are minimal attractors of $\Sigma_R$, $\psi_5$, $\psi_6$, $\psi_7$, $\psi_8$ (mixed) are unstable. 
\item[3)] There is a neighborhood $\mathcal{U}\setminus\Gamma$ of $0$ where $\Gamma$ is the stable manifold of $0$ such that $\bar{\mathcal{U}}=\cup_{i=1}^4 \bar{\mathcal{U}_i}$ with $\mathcal{U}_i$ pairwise disjoint and $\mathcal{U}_i$ is the basin of attraction of $\psi_i$, $i=1,\dots,4$.
\item[4)]  The projection of $\mathcal{U}_i$ onto the space spanned by $\phi_I$, $\phi_J$ is approximately a sectorial region given by:\[ \mathcal{U}_i\cap\{X\phi_I+Y\phi_J\mid \omega_{i,1}<arg(X,Y)<\omega_{i,2}\}, \quad i=1,\dots,4,\]
\begin{align*}
& \omega_{1,1}=\pi-\omega, &&\omega_{1,2}=\pi+\omega, &&\omega_{2,1}=-\omega, &&\omega_{2,2}=\omega, \\
&  \omega_{3,1}=\pi+\omega, &&\omega_{3,2}=2\pi-\omega, &&\omega_{4,1}=\omega,  &&\omega_{4,2}=\pi-\omega,
\end{align*}
where $\omega=\arctan{\sqrt{\frac{e-a}{e-f}}}$.
\end{itemize} 
\item[ii)] If $e<a$ and $e<f$ then the topological structure of the system after the transition is as in Figure \ref{eaf region}. In particular: 
\begin{itemize}
\item[a)] $\Sigma_R$ contains eight steady states $\psi_i$, $i=1,\dots,8$. 
\item[b)] $\psi_5$, $\psi_6$, $\psi_7$, $\psi_8$ (mixed) are minimal attractors of $\Sigma_R$ whereas $\psi_1$, $\psi_2$, $\psi_3$, $\psi_4$ (rectangles) are unstable.
\item[c)] There is a neighborhood $\mathcal{U}\setminus\Gamma$ of $0$ where $\Gamma$ is the stable manifold of $0$ such that $\bar{\mathcal{U}}=\cup_{i=5}^8 \bar{\mathcal{U}_i}$ with $\mathcal{U}_i$ pairwise disjoint and $\mathcal{U}_i$ is the basin of attraction of $\psi_i$, $i=5,\dots,8$. 
\item[d)] The projection of $\mathcal{U}_i$ onto the space spanned by $\phi_I$, $\phi_J$ is a sectorial region given by:
\[ \mathcal{U}_i\cap\{X\phi_I+Y\phi_J\mid \omega_{i,1}<arg(X,Y)<\omega_{i,2}\}, \quad i=5,\dots,8,\]
\begin{align*}
& \omega_{5,1}=0,&&\omega_{5,2}=\pi/2, \quad &&\omega_{6,1}=3\pi/2, && \omega_{6,2}=2\pi, \\
&  \omega_{7,1}=\pi,&&  \omega_{7,2}=3\pi/2, &&\omega_{8,1}=\pi/2,&& \omega_{8,2}=\pi.
\end{align*}

\end{itemize} 
\item[iii)] If $f<e<a$ then the topological structure of the system after the transition is as in Figure \ref{fea region}. In particular:
\begin{itemize}
\item[a)] $\Sigma_R$ contains four steady states $\psi_i$, $i=1,\dots,4$.
\item[b)] $\psi_1$, $\psi_2$ (rectangles) are minimal attractors of $\Sigma_R$ whereas $\psi_3$, $\psi_4$ (rectangles) are unstable steady states. 
\item[c)] There is a neighborhood $\mathcal{U}\setminus\Gamma$ of $0$ where $\Gamma$ is the stable manifold of $0$ such that $\bar{\mathcal{U}}=\cup_{i=1}^2 \bar{\mathcal{U}_i}$ with $\mathcal{U}_i$ pairwise disjoint and $\mathcal{U}_i$ is the basin of attraction of $\psi_i$, $i=1,2$. 
\item[d)] The projection of $\mathcal{U}_i$ onto the space spanned by $\phi_I$, $\phi_J$ is a sectorial region given by:
\[ \mathcal{U}_i\cap\{X\phi_I+Y\phi_J\mid \omega_{i,1}<arg(X,Y)<\omega_{i,2}\}, \quad i=1,2,\]
\[
\omega_{1,1}=\pi/2,\,  \omega_{1,2}=3\pi/2, \quad \omega_{2,1}=-\pi/2,\, \omega_{2,2}=\pi/2. 
\]
\end{itemize}
\item[iv)] If $a<e<f$ then the topological structure of the system after the transition is as in Figure \ref{aef region}. In particular:
\begin{itemize}
\item[a)] $\Sigma_R$ contains four steady states $\psi_i$, $i=1,\dots,4$.
\item[b)] $\psi_3$, $\psi_4$ (rectangles) are minimal attractors of $\Sigma_R$ whereas the $\psi_1$, $\psi_2$ (rectangles) are unstable steady states. 
\item[c)] There is a neighborhood $\mathcal{U}\setminus\Gamma$ of $0$ where $\Gamma$ is the stable manifold of $0$ such that $\bar{\mathcal{U}}=\cup_{i=3}^4 \bar{\mathcal{U}_i}$ with $\mathcal{U}_i$ pairwise disjoint and $\mathcal{U}_i$ is the basin of attraction of $\psi_i$, $i=3,4$. 
\item[d)] The projection of $\mathcal{U}_i$ onto the space spanned by $\phi_I$, $\phi_J$ is a sectorial region given by:
\[ \mathcal{U}_i\cap\{X\phi_I+Y\phi_J\mid \omega_{i,1}<arg(X,Y)<\omega_{i,2}\}, \quad i=3,4,\]
\[
 \omega_{3,1}=\pi,\,  \omega_{3,2}=2\pi, \quad \omega_{4,1}=0,\, \omega_{4,2}=\pi. 
\]
\end{itemize}
\end{itemize} \end{theorem}
 
\section{Proof of the Main Theorems}
First we give the preliminary setting that will be used in the proof of the main theorems.

The first step is to find the adjoint critical eigenvectors.  The adjoint equation of \eqref{EV} is:
\begin{equation}\label{adjoint}
\begin{aligned}
 & \Pr(\Delta \mathbf{u}^{\ast}-\nabla p^{\ast})+\theta^{\ast}{\bf k}=\bar{\beta}\mathbf{u}^{\ast},\\
 & \Delta\theta^{\ast}+R\, \Pr\, w^{\ast}=\bar{\beta}\theta^{\ast}.
\end{aligned}
\end{equation}
The eigenfunctions of \eqref{adjoint} can be represented by the separation of variables \eqref{sepofvar}. Also the eigenvalues of \eqref{adjoint} are same as the eigenvalues of \eqref{EV}, i.e. $\bar{\beta}$ satisfies \eqref{charequ}.  We find the amplitudes of the critical adjoint eigenvectors as:
\begin{equation}
W_S^{\ast}=\beta_S^1(R)+\gamma_S^2,\qquad \Theta_S^{\ast}=R\, \Pr, 
\end{equation}
where $\beta_S^1(R)$  satisfies the PES condition \eqref{PES1}. 

Let $I$ and $J$ be the indices of the critical modes, i.e. $\mathcal{C}=\{I,J\}$ in \eqref{PES1}. We will denote:
\[ \phi_I=\phi_I^1, \quad \phi_J=\phi_J^1, \quad \beta(R)=\beta_I^1(R)=\beta_J^1(R).\]  

We study the dynamics on the center manifold, we write:
\[\phi=\phi^{c}+\Phi(x),\] 
in \eqref{functional}, where $\Phi$ is the center manifold function and 
 \[\phi^{c}=x_1 \phi_{I}+x_2 \phi_{J}.\] 
 
Multiplying the governing evolution equation \eqref{functional} by the adjoint eigenvectors $\phi_{I}^{\ast}$ and $\phi_J^{\ast}$, we see that the amplitudes of the critical modes satisfy the following equations
\begin{equation}\label{reduced1}
\begin{aligned}
& \frac{dx_1}{dt}= \beta(R)x_1+\frac{1}{\left\langle \phi
_{I},\phi _{I}^{\ast }\right\rangle }\left\langle G( \phi ),\phi _{I}^{\ast }\right\rangle, \\
& \frac{dx_2}{dt}=\beta(R)x_2+\frac{1}{\left\langle \phi
_{J},\phi _{J}^{\ast }\right\rangle }\left\langle G( \phi ),\phi _{J}^{\ast }\right\rangle .
\end{aligned}
\end{equation}
The pairing $\left\langle \cdot ,\cdot \right\rangle $
denotes the $L^{2}$-inner product over $\Omega $.

By \eqref{PES1} and \eqref{charequ}, 
\begin{equation}\label{betaRc}
\beta(R)=g(R-R_c)+o((R-R_c)^2), \quad \text{as } R\rightarrow R_c,
\end{equation}
with 
\[g=\frac{\Pr \alpha^2}{(\Pr+1)\gamma^4}\neq 0.\]

We write the phase space as 
\begin{equation*}
H=E_{1}\oplus E_{2}, \qquad E_{1}=\text{span}\{\phi_{I},\phi_{J}\}, \qquad E_{2}=E_{1}^{\perp }.
\end{equation*}

When the linear part of \eqref{reduced1} is diagonal, we have the following approximation of the center manifold; see Ma and Wang \cite{ptd}:
\begin{equation}\label{cmrealformula}
-\mathcal{L}_{R}\Phi \left( x,R\right) =P_{2}G\left( \phi
^c\right) +o(2),
\end{equation}
where $\mathcal{L}_{R}=L_{R}\mid _{E_{2}}$, $P_2$ is the projection onto $E_2$ and
\begin{equation*}
o(n)=o\left( \left\vert x\right\vert ^{n}\right) +O\left( \left\vert
x\right\vert ^{n}\left\vert \beta\left( R\right) \right\vert
\right) .
\end{equation*}

Using \eqref{cmrealformula} we can write the center manifold function as:
\[
\Phi(x,y)=x^2\Phi_{1}+ xy\Phi_{2}+y^2 \Phi_{3}+o(2)
\]
Using this approximation and the bilinearity of $G$, we can write \eqref{reduced1} as:
\begin{equation}\label{reduced 2}
\begin{aligned}
&
\frac{dx}{dt}=\beta (R) x+ x(a_1 x^2+a_2 y^2)+o(3), \\
&
\frac{dy}{dt}=\beta (R) y +y(b_1 x^2+ b_2 y^2 )+o(3).
\end{aligned}
\end{equation}

In dynamic transition problems, the center manifold is generally expanded using the eigenfunctions of the original linear operator. However, following Sengul and Wang \cite{Sengul2011}, we will expand the center manifold using a different basis. Namely we will consider the eigenfunctions $\mathbf{u}_S$ of the Stokes equation for the velocity together with eigenfunctions $\theta_S$ of the Laplace equation. The main advantage of such an expansion is that the eigenvalues and eigenfunctions are independent of the system parameters, namely the Prandtl number $\Pr$ and the Rayleigh number $R$, while still spanning the same functional space \eqref{func spaces} which leads to computational advantages.

For this reason we turn to the following eigenvalue problem with the boundary conditions \eqref{bc} of the original problem:
\[
\begin{aligned}
& \Delta \mathbf{u}_S-\nabla p=\rho\, \mathbf{u}_S, \\
& \Delta \theta_S=\rho\, \theta_S,\\
& \text{div} \mathbf{u}_S=0.
\end{aligned}
 \]
By the classical theory of elliptic operators, the eigenvectors $\{e_S^1=(\mathbf{u}_S,0),\, e_S^2=(0,\theta_S)\}$ form a basis of the phase space $H$. Moreover, $e_S$ can be expressed by the same separation of variables \eqref{sepofvar}. There are three cases to be considered.  
 \begin{itemize}
 \item If $(s_x,s_y)=(0,0)$ and $s_z\neq0$, then $e_S^1=0$ and
 $$e_S^2=(0,\theta_S), \qquad \Theta_S=1.$$
 \item If $s_x^2+s_y^2\neq 0$ and $s_z=0$, then there are eigenmodes which have the form $e=(\mathbf{u},0)$ with $\mathbf{u}=(u,v,0)$. For such modes, it can be easily verified that:
 \[ \langle G(\phi^c),e\rangle=0. \]
 Thus by \eqref{cmrealformula}, such modes will not be present in the lowest order approximation of the center manifold function. 
 
\item  If $s_x^2+s_y^2\neq 0$ and $s_z\neq0$, then the multiplicity of an eigenvalue is two and the eigenvectors are:
\begin{equation*}
\begin{aligned}
&e_S^1=(\mathbf{u}_S,0), \qquad W_S=1,\\
&e_S^2=(0,\theta_S), \qquad \Theta_S=1.\\
\end{aligned}
\end{equation*}
\end{itemize}

The following lemma guarantees that the projection $P_2$ in \eqref{cmrealformula} can be ignored from a computational point of view.
\begin{lemma}\label{expansion lemma}
For $i=1,2$,
\begin{equation}\label{method1}
\begin{aligned}
& P_2 e_S^i=e_S^i \quad &&\text{for } S\notin \mathcal{C}, \\
& \langle P_2 G(\phi^c),e_S^i \rangle=0 && \text{for } S\in\mathcal{C}.
\end{aligned}
\end{equation}
\end{lemma}
\begin{proof}
First note that:
\[ E_1=\text{span}\{\phi_I^1,\phi_J^1\}\subset\text{span}\{e_S^1,e_S^2\mid S\in\mathcal{C}\}.\] 
Thus
\[E_2=E_1^{\perp}\supset \text{span}\{e_S^1,e_S^2\mid S\notin \mathcal{C}\}.\]
Thus we have the first equation in \eqref{method1}.

Since 
\[ 
\text{span}\{\phi_S^1,\phi_S^2\mid S\in\mathcal{C}\}=\text{span}\{e_S^1,e_S^2\mid S\in\mathcal{C}\},
\] 
there must exist constants $c_{S,1}\neq0$, $c_{S,2}\neq0$ such that
\[ P_2 e_S^1=c_{S,1} \phi_S^2, \quad P_2 e_S^2=c_{S,2} \phi_S^2, \quad \text{for } S\in\mathcal{C}.\]
Since by direct computation we have
\[ \langle P_2 G(\phi^c),\phi_S^2\rangle=0, \quad \text{for } S\in\mathcal{C}, \]
we also have second equation in \eqref{method1}.
\end{proof}

Now we write:
\begin{equation}
\Phi=\sum_{S\in\mathcal{S},\, i=1,2}\Phi_S^i(x,y) e_S^i+o(2),
\end{equation}
where $\mathcal{S}$ denotes some index set which will be specified later. Here $\Phi_S^i$  are  quadratic polynomials in $x$ and $y$. 

Let
\[
\begin{aligned}
& \mathcal{Z}_{\alpha}^{roll}=\{K=(k,0,1) \text{ or } K=(0,k,1): k\neq0,\,\alpha_K=\alpha\},\\
& \mathcal{Z}_{\alpha}^{rec}=\{K=(k_1,k_2,1):k_1\neq0,\,k_2\neq0, \alpha_K=\alpha\}.
\end{aligned}
\]

\begin{lemma}\label{S1Lemma}
For $S_1=(0,0,2)$, we have:
\[
\Phi_{S_1}=[\Phi_{S_1}^1,\Phi_{S_1}^2]^T=[0,\Phi_{S_1}^I x^2+\Phi_{S_1}^J y^2]^T,
\]
where for $K\in\{I,J\}$,
\begin{equation} \label{PhiS1}
\Phi_{S_1}^K=\left\{
\begin{array}{c}
-\frac{\gamma^2}{16\pi}[ 0 , 1]^T, \quad \text{if } K\in \mathcal{Z}_{\alpha}^{rec},\\
-\frac{\gamma^2}{8\pi}[ 0 , 1]^T, \quad \text{if } K\in \mathcal{Z}_{\alpha}^{roll}.
\end{array}\right.
\end{equation}
Also for $K_1\in\mathcal{Z}_{\alpha}^{rec},\, K_2\in\mathcal{Z}_{\alpha}^{roll}$:
\begin{equation}\label{GsJS1J}
 G_s(\phi_{K_1},e_{S_1},\phi_{K_1}^{\ast})=\frac{1}{2}G_s(\phi_{K_2},e_{S_1},\phi_{K_2}^{\ast})= \frac{L_1L_2\pi\gamma^2 R_c}{8}  [ 0, \Pr]^T.
\end{equation}
\end{lemma}
\begin{proof}
 If $(s_x,s_y)=(0,0)$ and $s_z\neq0$, then by \eqref{cmrealformula}, for $K\in\{I,J\}$, we have
  \begin{equation}\label{cm0}
  \begin{aligned}
 \Phi_{S}^K= \frac{\langle G( \phi_K,\phi_K),e_S^2\rangle}{\langle e_S^2,\mathcal{L}_R^{\ast} e_S^2\rangle}=\frac{\langle G( \phi_K,\phi_K),e_S^2\rangle}{-s_z^2 \pi^2\langle e_S^2,e_S^2\rangle}.
 \end{aligned}
 \end{equation}
 Note that
\begin{equation}\label{es2-es2}
\langle e_S^2,e_S^2\rangle=\int_{\Omega}\sin^2 2\pi z=\frac{L_1 L_2}{2}.
\end{equation}
For $K_1\in\mathcal{Z}_{\alpha}^{rec},\, K_2\in\mathcal{Z}_{\alpha}^{roll}$, a direct computation yields:
\begin{equation} \label{GIIS1}
\begin{aligned}
& G(\phi_{K_1},\phi_{K_1},e_{S_1})=\frac{1}{2}G(\phi_{K_2},\phi_{K_2},e_{S_1})=\frac{-L_1L_2\pi \gamma^2}{8} [ 0, 1]^T.
\end{aligned}
\end{equation}
Here $G_s$ is the trilinear operator defined in \eqref{bilinear}. Now,
\eqref{PhiS1} follows from \eqref{cm0}, \eqref{es2-es2} and \eqref{GIIS1}. Also, \eqref{GsJS1J} follows from an easy computation.
\end{proof}

\begin{lemma}
If $S\notin\mathcal{C}$, $S=(s_x,s_y,s_z)$, $(s_x,s_y)\neq(0,0)$ and $s_z\neq0$, then
\begin{equation} \label{cm1}
\Phi_S= \left[ \begin{array}{cc}
\Phi_{S}^1 \\ \Phi_S^2 \end{array}\right]= -\frac{1}{\mathfrak{v}_S}\mathcal{A_S}^{-1}\left[ \begin{array}{cc}
\langle G( \phi^{c}),e_S^1\rangle \\ \langle G( \phi^{c}),e_S^2\rangle \end{array}\right].
\end{equation}
Here
\begin{equation}\label{As}
 \mathcal{A}_S=
\left(\begin{array}{cc}
-\Pr\,R_S \gamma_S^{-2}  & R \Pr \\
1 & -\gamma_S^2  \end{array} \right),
\end{equation}
and
\begin{equation}\label{vs}
\mathfrak{v}_S=\left\{\begin{array}{c}
\frac{L_1L_2}{4}, \text{if } s_x s_y=0,\\
\frac{L_1L_2}{8}, \text{if } s_x s_y\neq0.
\end{array} \right. 
\end{equation}
\begin{equation}\label{Rs}
R_S=\frac{\gamma_S^6}{\alpha_S^2}.
\end{equation}
\end{lemma}

\begin{proof}
Let
 \begin{equation}\label{CM3}
\mathcal{A}_S^{nm}=\frac{1}{\mathfrak{v}_S} \langle e_S^m,\mathcal{L}_R^{\ast} e_S^n\rangle, \qquad m,n=1,2.
 \end{equation}
That is:
\[
\mathcal{A}_S= \frac{1}{\mathfrak{v}_S} \left( \begin{array}{cc}
 \langle e_S^1,\mathcal{L}_R^{\ast} e_S^1\rangle & \langle e_S^2,\mathcal{L}_R^{\ast} e_S^1\rangle \\
 \langle e_S^1,\mathcal{L}_R^{\ast} e_S^2\rangle & \langle e_S^2,\mathcal{L}_R^{\ast} e_S^2\rangle
 \end{array}\right).
 \]
Let
\begin{equation*}
\mathfrak{v}_S=\int_{0}^{L_1}\int_{0}^{L_2}\int_{0}^{1}\cos^2\frac{s_x\pi x_1}{L_1}\cos^2\frac{s_y\pi x_2}{L_2}\cos^2s_z\pi x_3.
\end{equation*}
Clearly $\mathfrak{v}_S$ is equal to the definition in \eqref{vs}. Also it is easy to compute that $\mathcal{A}_S$ is the same as given in \eqref{As}.
Now the lemma can be proved by taking the inner product of \eqref{cmrealformula} by $e_S^{i}$, $i=1,2$ and using Lemma~\ref{expansion lemma}.
 \end{proof}
Notice that 
\[\det \mathcal{A}_S=\mathfrak{v}_S^2 Pr(R_S-R).
\]
Since $R_S>R_c$ for $S\notin\mathcal{C}$, the determinant of $\mathcal{A}$ is always positive when $R$ is close to $R_c$. This guarantees that \eqref{cm1} can be solved for $\Phi_S^1$ and $\Phi_S^2$.

Finally note that for $S=(s_x,s_y,s_z)$ with $(s_x,s_y)\neq(0,0)$ and $s_z\neq0$ at $R=R_c$ we have:
\begin{equation}\label{normpsi}
\begin{split}
\langle \phi_{S},\phi_{S}^{\ast} \rangle &= \int_{\Omega}u_Su^{\ast}_S+v_Sv^{\ast}_S+w_Sw^{\ast}_S+\theta_S\theta^{\ast}_S\\
&=\mathfrak{v}_S\left((\frac{s_z^2\pi^2}{\alpha_S^2}+1)W_SW_S^{\ast}+\Theta_S\Theta_S^{\ast}\right)\\
&=\mathfrak{v}_S(\frac{\gamma_S^2}{\alpha_S^2}(\gamma_S^2+\beta(R))^2+R\,\Pr)\mid_{R=R_c}\\
&=\mathfrak{v}_S(\Pr+1)\gamma_S^4.
\end{split}
\end{equation}

\begin{notation}
We will use the following notation:
\[G(\phi_I,\phi_J,e_{S})=\left[ \begin{array}{c} G(\phi_I,\phi_J,e_{S}^1) \\G(\phi_I,\phi_J,e_{S}^2) \end{array}\right]. \]

We also define the following indices which will be used through out the proofs.
\begin{equation} \label{Si}
\begin{split}
& S_1=(0,0,2), \,S_2=(2i_x,0,2), \,S_3=(0,2i_y,2), \,S_4=(i_x,i_y+j_y,2), \\
&S_5=(i_x,j_y-i_y,2), \,S_6=(i_x,j_y,2), \,S_7=(i_x+j_x,i_y+j_y,2), \\ 
&S_8=(i_x-j_x,i_y+j_y,2),\,S_9=(i_x+j_x,-i_y+j_y,2),\\
& S_{10}=(i_x-j_x,-i_y+j_y,2), \,S_{11}=(2j_x,0,2), \,S_{12}=(0,2j_y,2).
\end{split}
\end{equation}
\end{notation}

\subsection{Proof of Theorem \ref{roll vs rec thm}}
Now we assume that $I=(i_x,i_y,1)$ and $J=(0,j_y,1)$, ($i_x\geq1$, $j_y>i_y\geq1$) are the first critical indices with identical wave numbers, $\alpha_I=\alpha_J$. 
Using \eqref{cmrealformula}, we find that the lowest order approximation of the center manifold is spanned by the eigenvectors having indices $S_{i}$ as given in \eqref{Si}: 
\[
\Phi(x,y)=
x^2(\Phi_{S_1}^I e_{S_1}+\sum_{S=S_2,S_3} \Phi_{S}e_S)+xy\sum_{S=S_4,S_5} \Phi_{S}e_S+y^2 \Phi_{S_1}^J e_{S_1}.
\]
Here
\[ \Phi_{S_1}=x^2\Phi_{S_1}^I+y^2\Phi_{S_1}^J,\]
where according to \eqref{PhiS1} we have:
\[ \Phi_{S_1}^J=2\Phi_{S_1}^I.\]
The coefficients in \eqref{reduced 2} are as follows:
\begin{equation}\label{roll vs rec coefs}
\begin{aligned}
& a_1=\frac{1}{\langle \phi_I,\phi_I^{\ast}\rangle}\left(\Phi_{S_1}^I G_s(\phi_I,e_{S_1},\phi_I^{\ast})+\sum_{i=2,3} \Phi_{S_i}G_s(\phi_I,e_{S_i},\phi_I^{\ast})\right), \\
& a_2=\frac{1}{\langle \phi_I,\phi_I^{\ast}\rangle}\left(\Phi_{S_1}^J G_s(\phi_I,e_{S_1},\phi_I^{\ast})+\sum_{i=4,5} \Phi_{S_i}G_s(\phi_J,e_{S_i},\phi_I^{\ast})\right), \\
& b_1=\frac{1}{\langle \phi_J,\phi_J^{\ast}\rangle}\left(\Phi_{S_1}^I G_s(\phi_J,e_{S_1},\phi_J^{\ast})+\sum_{i=4,5} \Phi_{S_i}G_s(\phi_I,e_{S_i},\phi_J^{\ast})\right), \\
& b_2=\frac{1}{\langle \phi_J,\phi_J^{\ast}\rangle}\Phi_{S_1}^J G_s(\phi_J,e_{S_1},\phi_J^{\ast}).
\end{aligned}
\end{equation}
All the terms involving $S_1$ above can be computed using the Lemma~\ref{S1Lemma}.

By direct computation:
\begin{equation}\label{GsISJ}
G_s(\phi_I,e_S,\phi_J^{\ast})=G_s(\phi_J,e_S,\phi_I^{\ast}),  \quad S=S_4,\, S_5.
\end{equation}
Putting \eqref{PhiS1}, \eqref{GsISJ} and \eqref{normpsi} into \eqref{roll vs rec coefs} we find that:
\begin{equation} \label{a2=2b1}
a_2=2b_1.
\end{equation}
By direct computation we can obtain the following:
\begin{equation}\label{res1}
\begin{aligned}
& G(\phi_I,\phi_I,e_{S})=-\frac{1}{2}\eta_S \left[ \begin{array}{c} \alpha^{-2} \\ \gamma^{-4} \end{array}\right], && S=S_2,S_3,\\
& G_s(\phi_I,\phi_J,e_S)=-\eta_S \left[ \begin{array}{c} \alpha^{-2} \\ \gamma^{-4} \end{array}\right], && S=S_4,S_5,
\end{aligned}
\end{equation}
\begin{equation}\label{res2}
\begin{aligned}
& G_s(\phi_I,e_{S},\phi_I^{\ast})=\frac{1}{2}\eta_S\left[ \begin{array}{c} \alpha^{-2} \\\Pr\,R_c \gamma^{-4} \end{array}\right], && S=S_2,S_3, \\
&G_s(\phi_J,e_S,\phi_I^{\ast})=\frac{1}{2}\eta_S\left[ \begin{array}{c} \alpha^{-2} \\\Pr\,R_c \gamma^{-4} \end{array}\right], && S=S_4, S_5,
\end{aligned}
\end{equation}
where
\begin{equation}\label{eta}
\eta_S=\frac{L_1 L_2 \pi(4\alpha^2-\alpha_S^2) \gamma^6}{32\alpha^2}.
\end{equation}
Using \eqref{cm0}, \eqref{cm1} and \eqref{res1} , we find:
\begin{equation}\label{res3}
\begin{aligned}
& \Phi_{S}=\left[ \begin{array}{c} \Phi_S^1\\\Phi_S^2\end{array}\right]=\frac{1}{16}\frac{\pi (4\alpha^2-\alpha_S^2) \gamma^2}{\Pr\, \alpha^4(R_c-R_S)}\left[ \begin{array}{c} \Pr \,R_c \alpha^2+\gamma^4\gamma_S^2\\ \Pr\, \frac{R_S}{\gamma_S^2}\alpha^2+\gamma^4\end{array}\right], \quad S=S_2,S_3,\\
& \Phi_{S}=\left[ \begin{array}{c} \Phi_S^1\\\Phi_S^2\end{array}\right]=\frac{1}{4}\frac{\pi (4\alpha^2-\alpha_S^2) \gamma^2}{\Pr\, \alpha^4(R_c-R_S)}\left[ \begin{array}{c} \Pr \,R_c \alpha^2+\gamma^4\gamma_S^2\\ \Pr\, \frac{R_S}{\gamma_S^2}\alpha^2+\gamma^4\end{array}\right], \quad S=S_4,S_5.\\
\end{aligned}
\end{equation}
Now putting \eqref{res2}, \eqref{res3}, \eqref{GIIS1}, \eqref{PhiS1} into \eqref{roll vs rec coefs} and normalizing the results by
\begin{equation}\label{normalizer}
\frac{L_1 L_2 R_c\gamma^4}{2^{10}\Pr \,\alpha^2\langle \phi_I,\phi_I^{\ast}\rangle}=\frac{ R_c}{2^{7}\alpha^2\Pr (1+\Pr)},
\end{equation}
and using \eqref{a2=2b1}, the equations \eqref{reduced 2} become:
\begin{equation}\label{roll vs rec reduced}
\begin{aligned}
&
\frac{dx}{dt}=\beta (R) x- x(a x^2+2c y^2)+o(3), \\
&
\frac{dy}{dt}=\beta (R) y -y(c x^2+ 2b y^2 )+o(3),
\end{aligned}
\end{equation}
where $a$, $b$ and $c$ are as defined in \eqref{definition abc}.
Note that $\kappa_{S_i}>0$ in \eqref{kappa} since $R_c<R_{S_i}$, $i=1,\dots,5$ and we have the following relations:
\begin{equation}\label{abc relation}
0<b<a,  \quad 0<b<c.
\end{equation}
Now we consider the approximate steady state equations of \eqref{roll vs rec reduced}:
\begin{equation}\label{roll vs rec approximate}
\begin{aligned}
& \beta (R) x- x(a x^2+2c y^2)=0, \\
& \beta (R) y -y(c x^2+ 2b y^2 )=0.
\end{aligned}
\end{equation}
Let us define 
\begin{equation}\label{d}
\omega^2:=\frac{c-a}{2(c-b)}.
\end{equation}
There are two cases to consider.
\begin{itemize}
\item[i)] If $c<a$, then the equations \eqref{roll vs rec approximate} have only four straight line orbits on the lines $y=0$ and $x=0$. And the following solutions of \eqref{roll vs rec approximate} are bifurcated on $\beta>0$:
\begin{equation}\label{steady1}
X_{\pm}=(\pm\sqrt{\frac{\beta}{a}},0), \quad Y_{\pm}=(0,\pm\sqrt{\frac{\beta}{2b}}).
\end{equation}
\item[ii)] If $c>a$  then there are four additional straight line orbits on the lines $y=\pm\omega x$. Note that in this case by \eqref{abc relation}, $c^2-ab>0$ and there are four additional  solutions bifurcated on $\beta>0$ which are given by:
\begin{equation}\label{steady2}
Z_{\pm}^{i}=(-1)^i(1,\pm \omega)\sqrt{\beta\frac{c-b}{c^2-a b}}, \quad i=1,2.
\end{equation}
\end{itemize}
Now the Jacobian of the vector field in \eqref{roll vs rec approximate} is:
\begin{equation}
J=\left(
\begin{array}{cc}
\beta-3a x^2-2c y^2 & -4c y x \\
-2c y x & \beta-c x^2-6b y^2
\end{array}
\right).
\end{equation}
The Jordan canonical forms of the Jacobian matrix evaluated at the steady states in \eqref{steady1} and \eqref{steady2} are as follows:
\begin{equation}\label{roll vs rec e.values}
\begin{aligned}
& J(X_{\pm})=2\beta\left(
\begin{array}{cc}
-1 & 0 \\
0 & \frac{a-c}{2a}
\end{array}
\right), 
\quad
J(Y_{\pm})=2\beta\left(
\begin{array}{cc}
\frac{b-c}{2b} & 0 \\
0 & -1
\end{array}
\right), 
\\
& J(Z_{\pm}^i)\sim 2\beta\left(
 \begin{array}{cc}
-1 & 0\\
0 & \frac{(b-c)(c-a)}{ab-c^2}
\end{array}
\right), \quad i=1,2.
\end{aligned}
\end{equation}
The stability of the steady states can be found by using \eqref{roll vs rec e.values}. Since $c>b$,  $Y_{\pm}$ are always stable on $R>R_c$. If $c>a$ then $X_{\pm}$ are stable and $Z_{\pm}$ are unstable on $R>R_c$. On the other hand if $c<a$ then $X_{\pm}$ are unstable and $Z_{\pm}$ are not bifurcated on $R>R_c$. Thus we have two transition scenarios and the results are shown in Figure \ref{regions}.

\begin{remark}\label{Rem}In the particular case $j_y=2i_y$, the equations in \eqref{roll vs rec reduced} can be reduced further. In this case we have: $\alpha_{S_2}=\alpha_{S_4}$ and $\alpha_{S_3}=\alpha_{S_5}$ which implies that $\gamma_{S_2}=\gamma_{S_4}$ and $\gamma_{S_3}=\gamma_{S_5}$. This in turn implies $\kappa_{S_2}=\kappa_{S_4}$ and $\kappa_{S_3}=\kappa_{S_5}$ and we get the relation $2a=b+c$. This implies $c>a$ since $c-a=a-b>0$.
\end{remark}

\subsection{Proof of Theorem \ref{roll vs roll thm}}
We point out the differences from the previous proof. We have $I=(i_x,0,1)$ and $J=(0,j_y,1)$ ($i_x\geq1$, $j_y\geq1$) as the first critical indices with identical wave numbers, $\alpha_I=\alpha_J$. First the center manifold is given by:
\[
\Phi(x,y)=
x^2\Phi_{S_1}^I\phi_{S_1}+xy\Phi_{S_6}e_{S_6}+y^2 \Phi_{S_1}^J e_{S_1}+o(2).
\]
Using this approximation, the coefficients in \eqref{reduced 2} are as follows:
\begin{equation}\label{roll vs roll coef}
\begin{aligned}
& a_1=\frac{1}{\langle \phi_I,\phi_I^{\ast}\rangle}\Phi_{S_1}^I G_s(\phi_I,e_{S_1},\phi_I^{\ast}), \\
& a_2=\frac{1}{\langle \phi_I,\phi_I^{\ast}\rangle}\left(\Phi_{S_1}^J G_s(\phi_I,e_{S_1},\phi_I^{\ast})+\Phi_{S_6}G_s(\phi_J,e_{S_6},\phi_I^{\ast})\right), \\
& b_1=\frac{1}{\langle \phi_J,\phi_J^{\ast}\rangle}\left(\Phi_{S_1}^I G_s(\phi_J,e_{S_1},\phi_J^{\ast})+\Phi_{S_6}G_s(\phi_I,e_{S_6},\phi_J^{\ast})\right), \\
& b_2=\frac{1}{\langle \phi_J,\phi_J^{\ast}\rangle}\Phi_{S_1}^J G_s(\phi_J,e_{S_1},\phi_J^{\ast}).
\end{aligned}
\end{equation}
The coefficients $a_1$ and $b_2$ are equal to $b_2$ in the proof of the first theorem. Thus we only need to find $a_2$ and $b_1$. A quick computation shows that:
\begin{equation}\label{roll vs roll GsISJ}
G_s(\phi_I,e_{S_6},\phi_J^{\ast})=G_s(\phi_J,e_{S_6},\phi_I^{\ast}),
\end{equation}
which shows that:
\begin{equation} \label{roll vs roll a2,b1}
a_2=b_1.
\end{equation}
By direct computation we can obtain the following:
\begin{equation}\label{roll vs roll res1}
G_s(\phi_I,\phi_J,e_{S})=-2\eta_S\left[ \begin{array}{c} \alpha^{-2} \\ \gamma^{-4} \end{array}\right], \quad S=S_6,
\end{equation}
\begin{equation}\label{roll vs roll res2}
G_s(\phi_J,e_{S},\phi_I^{\ast})=\eta_S \left[ \begin{array}{c} \alpha^{-2} \\\Pr \,R_c \gamma^{-4}\end{array}\right], \quad S=S_6.
\end{equation}
where $\eta$ is given by \eqref{eta}.
Using \eqref{cm0}, \eqref{cm1} and \eqref{roll vs roll res1} , we find:
\begin{equation}\label{roll vs roll res3}
\Phi_{S}=\left[ \begin{array}{c} \Phi_S^1\\\Phi_S^2\end{array}\right]=\frac{1}{2}\frac{\pi (4\alpha^2-\alpha_S^2) \gamma^2}{\Pr\, \alpha^4(R_c-R_S)}\left[ \begin{array}{c} \Pr \,R_c \alpha^2+\gamma^4\gamma_S^2\\ \Pr\, \frac{R_S}{\gamma_S^2}\alpha^2+\gamma^4\end{array}\right], \quad S=S_6.\\
\end{equation}
By normalizing all terms in \eqref{roll vs roll coef} by the normalizing factor \eqref{normalizer} and using \eqref{roll vs roll a2,b1}, the equations \eqref{reduced 2} become:
\begin{equation}\label{roll vs roll reduced}
\begin{aligned}
&
\frac{dx}{dt}=\beta (R) x- x(b x^2+d y^2)+o(3), \\
&
\frac{dy}{dt}=\beta (R) y -y(d x^2+ b y^2 )+o(3).
\end{aligned}
\end{equation}
Here $b>0$ and $d>0$ are definition \eqref{definition bd}. Since $\kappa>0$, we always have $b<d$.
Now we consider the approximate steady state equations of \eqref{roll vs roll reduced}:
\begin{equation}\label{roll vs roll approximate}
\begin{aligned}
& \beta (R) x- x(b x^2+d y^2)=0, \\
& \beta (R) y -y(d x^2+ b y^2 )=0.
\end{aligned}
\end{equation}
The equations \eqref{roll vs roll approximate} have always eight straight line orbits on the lines $y=0$, $x=0$ and $y=\pm x$. The eight solutions of \eqref{roll vs roll approximate} which are  bifurcated on $\beta>0$ are:
\begin{equation}\label{roll vs roll steady}
\begin{aligned}
& X_{\pm}=(\pm\sqrt{\frac{\beta}{b}},0), \quad Y_{\pm}=(0,\pm\sqrt{\frac{\beta}{b}}), \\
& Z_{\pm}^{i}=(-1)^i(1,\pm 1)\sqrt{\frac{\beta}{b+d}}, \quad i=1,2.
\end{aligned}
\end{equation}

The Jordan canonical forms of the Jacobian matrix evaluated at the steady states in \eqref{roll vs roll steady} are as follows:
\begin{equation}\label{roll vs roll e.values}
\begin{aligned}
& J(X_{\pm})=2\beta\left(
\begin{array}{cc}
-1 & 0 \\
0 & \frac{b-d}{2b}
\end{array}
\right), 
\quad
J(Y_{\pm})=2\beta\left(
\begin{array}{cc}
\frac{b-d}{2b} & 0 \\
0 & -1
\end{array}
\right), \\
& J(Z_{\pm}^i)\sim 2\beta\left(
 \begin{array}{cc}
-1 & 0\\
0 & \frac{d-b}{d+b}
\end{array}
\right), \quad i=1,2.
\end{aligned}
\end{equation}
Using \eqref{roll vs roll e.values} we can find that, since $b<d$, $X_{\pm}$ and $Y_{\pm}$ are stable and $Z_{\pm}^i$, $i=1,2$ are unstable on $R>R_c$ .
That finishes the proof.

\subsection{Proof of Theorem \ref{rec vs rec thm}}
Again we point out the differences from the previous proofs. We have $I=(i_x,i_y,1)$, $J=(j_x,j_y,1)$ ($i_x\neq0$, $i_y\neq0$, $j_x\neq0$, $j_y\neq0$) as the first critical indices with identical wave numbers, $\alpha_I=\alpha_J$. First the center manifold is given by:
\[
\begin{split} 
\Phi(x,y)=
&x^2(\Phi_{S_1}^I\phi_{S_1}+\sum_{S=S_2,S_3} \Phi_{S}e_S)+xy\sum_{S=S_7,S_8,S_9,S_{10}} \Phi_{S}e_S\\
&+y^2(\Phi_{S_1}^J\phi_{S_1}+\sum_{S=S_2,S_3} \Phi_{S}e_S)+o(2).
\end{split}
\]
Using this approximation, the coefficients in \eqref{reduced 2} are as follows:
\begin{equation}\label{rec vs rec coef}
\begin{aligned}
& a_1=\frac{1}{\langle \phi_I,\phi_I^{\ast}\rangle}(\Phi_{S_1}^I G_s(\phi_I,e_{S_1},\phi_I^{\ast})+\sum_{i=2,3} \Phi_{S_i}G_s(\phi_I,e_{S_i},\phi_I^{\ast})), \\
& a_2=\frac{1}{\langle \phi_I,\phi_I^{\ast}\rangle}(\Phi_{S_1}^J G_s(\phi_I,e_{S_1},\phi_I^{\ast})+\sum_{i=7,\dots,10} \Phi_{S_i}G_s(\phi_J,e_{S_i},\phi_I^{\ast})), \\
& b_1=\frac{1}{\langle \phi_J,\phi_J^{\ast}\rangle}(\Phi_{S_1}^I G_s(\phi_J,e_{S_1},\phi_J^{\ast})+\sum_{i=7,\dots,10} \Phi_{S_i}G_s(\phi_I,e_{S_i},\phi_J^{\ast})), \\
& b_2=\frac{1}{\langle \phi_J,\phi_J^{\ast}\rangle} (\Phi_{S_1}^J G_s(\phi_J,e_{S_1},\phi_J^{\ast})+\sum_{i=11,12} \Phi_{S_i}G_s(\phi_J,e_{S_i},\phi_J^{\ast}) ).
\end{aligned}
\end{equation}
$a_{1}$ and $b_{2}$ is computed in the same way as $a_1$ in the proof of the first theorem. So we will only find $a_2$ and $b_1$.
Also using
\begin{equation}\label{rec vs rec GsISJ}
G_s(\phi_I,e_{S},\phi_J^{\ast})=G_s(\phi_J,e_{S},\phi_I^{\ast}), \quad S=S_7,S_8,S_9,S_{10},
\end{equation}
we find that:
\begin{equation} \label{rec vs rec a2,b1}
a_2=b_1.
\end{equation}

By direct computation we can obtain the following:
\begin{equation}\label{rec vs rec res1}
G_s(\phi_I,\phi_J,e_{S})=-\frac{1}{2}\eta_S \left[ \begin{array}{c} \alpha^{-2} \\ \gamma^{-4} \end{array}\right], \quad S=S_7,S_8,S_9,S_{10},
\end{equation}

\begin{equation}\label{rec vs rec res2}
G_s(\phi_J,e_{S},\phi_I^{\ast})=\frac{1}{4}\eta_S \left[ \begin{array}{c} \alpha^{-2} \\\Pr \,R_c\gamma^{-4}\end{array}\right], \quad S=S_7,S_8,S_9,S_{10}.
\end{equation}
where $\eta$ is given by \eqref{eta}.
Using \eqref{cm0}, \eqref{cm1} and \eqref{rec vs rec res1} , we find:
\begin{equation}\label{rec vs rec res3}
\Phi_{S}=\left[ \begin{array}{c} \Phi_S^1\\\Phi_S^2\end{array}\right]=\frac{1}{8}\frac{\pi (4\alpha^2-\alpha_S^2) \gamma^2}{\Pr\, \alpha^4(R_c-R_S)}\left[ \begin{array}{c} \Pr \,R_c \alpha^2+\gamma^4\gamma_S^2\\ \Pr\, \frac{R_S}{\gamma_S^2}\alpha^2+\gamma^4\end{array}\right], \quad S=S_7,S_8,S_9,S_{10}.\\
\end{equation}
By normalizing all terms in \eqref{rec vs rec coef} by the normalizing factor \eqref{normalizer} and using \eqref{rec vs rec a2,b1}, the equations \eqref{reduced 2} become:
\begin{equation}\label{rec vs rec reduced}
\begin{aligned}
&
\frac{dx}{dt}=\beta (R) x- x(a x^2+e y^2)+o(3), \\
&
\frac{dy}{dt}=\beta (R) y -y(e x^2+ f y^2 )+o(3).
\end{aligned}
\end{equation}
Here $a$, $e$ and $f$ are positive numbers defined in \eqref{definition aef}.

Now we consider the approximate steady state equations of \eqref{rec vs rec reduced}:
\begin{equation}\label{rec vs rec approximate}
\begin{aligned}
& \beta (R) x- x(a x^2+e y^2)=0, \\
& \beta (R) y -y(e x^2+ f y^2 )=0.
\end{aligned}
\end{equation}
The analysis of the equations \eqref{rec vs rec approximate} is similar to the analysis of the equations \eqref{roll vs rec approximate} given in the proof of Theorem \ref{roll vs rec thm}. Thus we omit the details.

\section{Physical Remarks}
In this section we will use the main theorems to derive some physical conclusions. 

Before going into details, we have to make a remark about the critical wave number $\alpha$ since it is one of the parameters determining the transition numbers. Although $\alpha$ depends on the length scales, thanks to Lemma \ref{wave number estimate lemma}, we have certain bounds on its range of values. To recall, these are:
\begin{itemize}
\item[i)] $\alpha>1.55$ regardless of the length scale.
\item[ii)] If one of the length scales is greater than $2.03$ then $\alpha<3.10$.
\end{itemize} 
Taking a look into {\it Figure \ref{wavenumberestimate}}, one sees that the only case which is not covered by taking $\alpha$ in the range $1.55<\alpha<3.10$ is that of two critical rolls with indices $I=(1,0,1)$ and $J=(0,1,1)$ which happens when $L_1=L_2<1.69$.

There are only three possible cases when two modes with equal wave numbers become unstable simultaneously. Namely these two critical modes can be, a rectangular and a roll mode, both roll modes, both rectangular modes. We investigate each case separately.
\subsection{The first two critical modes are a roll and a rectangle.}
We first consider two critical wave indices $I=(i_x,i_y,1)$ of a rectangular pattern and $J=(0,j_y,1)$ of a roll pattern with equal wave numbers $\alpha=\alpha_I=\alpha_J$; see {\it Figure  \ref{RollRecPat}}.  We define the following number:
 \[ 
 A=\frac{j_y}{i_y}.
\]
Since we are assuming that $\alpha_I=\alpha_J$, we have $1\leq i_y<j_y$ and $A>1$. Notice that $A$ is the number of rolls to the number of the rectangle columns in the direction of rolls.  {\it Figure  \ref{wavemap}} shows that the possible values of $A$ in the small length scale regime $\max\{L_1,L_2\}<5.5$  are $4/3$, $3/2$, $2$, $3$, $4$.

Now to use Theorem \ref{roll vs rec thm} to describe the pattern selection after the transition, we have to compute the transition numbers $a$ and $c$ given by \eqref{definition abc}. These numbers depend on three parameters $\Pr$, $L_1$ and $L_2$. Equivalently one can use $\Pr$, $A$ and $\alpha$ as the parameters determining $a$, $b$ and $c$. 

Using
\[
\alpha_{S_2}=2\sqrt{\alpha^2-\frac{\alpha^2}{A^2}}, \quad \alpha_{S_3}=2\frac{\alpha}{A}, \quad \alpha_{S_4}=\sqrt{2\alpha(\alpha+\frac{\alpha}{A})}, \quad \alpha_{S_5}=\sqrt{2\alpha(\alpha-\frac{\alpha}{A})},
\]
in \eqref{definition abc}, we can compute the numerical values of $a$ and $c$ for a given value of $A$, $\alpha$ and $\Pr$.

There are two cases that can happen, depending on whether $a<c$ or $c<a$. In {\it Figure \ref{acplot}}, the regions where $c<a$ is shown for several parameter regimes. The results show that both transition scenarios described by Theorem \ref{main} are possible. In particular there are two parameter regimes such that $c<a$, namely when $\Pr<0.2$ and $1<A<1.4$ and for large $A$ or large $\Pr$.

Now some remarks about the basin of attraction of the rolls and rectangles are in order. When $a<c$, both rectangles and rolls are stable but their basins of attraction which are sectorial regions depend on an angle $\omega$ which depends on $\Pr$, $\alpha$ and $A$. In the particular case $A=2$, the mixed modes have a hexagonal pattern and we find that $a<c$. Moreover, $\omega$ is independent of $\alpha$ and $\Pr$ and is found to be $\omega=\arctan 1/2\approx 26.57^{\circ}$. So in this case the basin of attraction of rolls consists of two sectors each of which has an angle of $\pi-2\omega\approx126.87^{\circ}$ while the basin for the rectangles have an angle of $2\omega\approx 53.13^{\circ}$. This means that rolls will attract a wider region of initial conditions than rectangles do.

\subsection{The first two critical modes are both rolls.}
Now we consider two roll type critical modes. By the assumption of equal wave  numbers, the rolls has to be perpendicular to each other, i.e. $I=(i_x,0,1)$ and $J=(0,j_y,1)$. In this case after the first dynamic transition, always the rolls are stable and the mixed states are unstable. Moreover, rolls with index $I$ and rolls with index $J$ have uniform attraction basins.

\subsection{The first two critical modes are both rectangles.}
Now we consider the case where the first two critical modes both have rectangle patterns, i.e. $I=(i_x,i_y,1)$ and $J=(j_x,j_y,1)$ ($i_x>j_x\geq1$, $j_y>i_y\geq1$) are the first critical wave indices. In this case, the dynamic transitions depend on the numbers $a$, $e$ and $f$ given by Theorem \ref{rec vs rec thm}. To calculate these numbers we define the parameters $A$ and $B$:
\[
A=\frac{j_y}{i_y}>1, \quad B=\frac{i_x}{j_x}>1.
\]
Using the definition, we find that:
\[
\begin{aligned}
& \alpha_{2i_x,0}^2=4\frac{B^2 (A^2-1)}{B^2A^2-1}\alpha^2, \quad \alpha_{0,2i_y}^2=4\frac{B^2-1}{B^2A^2-1}\alpha^2,\\
& \alpha_{2j_x,0}^2=4\frac{A^2 (B^2-1)}{B^2A^2-1}\alpha^2, \quad \alpha_{0,2j_y}^2=4\frac{A^2-1}{B^2A^2-1}\alpha^2,\\
&\alpha_{i_x+(-1)^m j_x,(-1)^ni_y+j_y}^2= 2\left(\alpha^2+\frac{(-1)^m}{4B}\alpha_{2i_x,0}^2+A\frac{(-1)^n }{4}\alpha_{0,2i_y}^2\right), \quad m,n=1,2.
\end{aligned}
\]
Using this, we  computed $a$, $e$ and $f$ for several choices of $\Pr$, $\alpha$, $B$ and $C$. Our numerical calculations revealed that $a<e$ and $f<e$ for a vast amount of parameter choices. This means that the transition scenario is described as {\it Figure \ref{afe region}}. Then the rectangles with index $I$ and $J$ are both stable and the mixed modes are unstable.

However, the other transition scenarios can also be possible as we observed that $f<e<a$ when the Prandtl number is small, one of $A$ or $B$ is less than 2 and $A\neq B$. For an example see {\it Figure \ref{recvsrectest}}. In this case the transition scenario is described as {\it Figure \ref{fea region}}. Hence only rectangles with index $I$ are stable and the rectangles with index $J$ are unstable.

In particular, our numerical investigations suggest that one or both of the pure modes (rectangles) are stable while the mixed modes are unstable.

\section{Conclusions}
In this paper, we discuss the dynamic transitions of Rayleigh B{\'e}nard (RB) convection from a perspective of pattern formation. We focus on the case when two eigenvalues cross the imaginary axis simultaneously. This allows us to compare the stability of a pattern with respect to perturbations of other pattern types.  Our main assumption is that the wave numbers of the critical modes are equal. Under this assumption, we classify all the possible transition scenarios and determine in each case the preferred patterns and their basins of attraction depending on the system parameters.

The pattern of a simple critical mode is either a rectangle or a roll. Thus there are three possible cases when there are two critical modes: (a) one mode is rectangular, the other mode is a roll, (b) both modes are rolls, (c) both modes are rectangles. 
  
The following are some general characteristics of the transition for the RB convection with which already known (Ma and Wang \cite{Ma2004a,Ma2007}): 
 \begin{itemize}
\item[1)] The transition is Type-I. In particular, there is an attractor $\Sigma_R$ bifurcating on $R>R_c$. 
\item[2)] $\Sigma_R$ is homeomorphic to $S^1$ which comprises steady states and the connecting heteroclinic orbits.
\end{itemize}
The following are the results due to our main theorems:
\begin{itemize}
\item[3)] In all the scenarios, we found that only pure modes (rolls or rectangles) are stable and the mixed modes are unstable. Our result is conclusive (analytical proof) when one of the critical modes is a roll type. When both critical modes are rectangles, we only have computational evidence. 
\item[4)] When both critical modes are rolls, the stable steady states after the transition are rolls. When both critical modes are rectangles, computational evidence suggests that the stable steady states after the transition are rectangles. When one critical mode is a roll and the other one is a rectangle, the stable states after the transition can be either only rolls or both rolls and rectangles. 
\item[5)] When both rolls and rectangles are stable after the transition, these states have non-uniform sectorial basin of attractions. In the particular case, where the mixed states have a regular hexagonal pattern, the angle of the sector for rolls is $126.87^{\circ}$ while the angle of the sector for the rectangles is $53.13^{\circ}$. Thus rolls attract a wider range of initial conditions, making them a more preferable type of pattern.
 \end{itemize}

\appendix

\begin{figure}
\includegraphics[scale=.58]{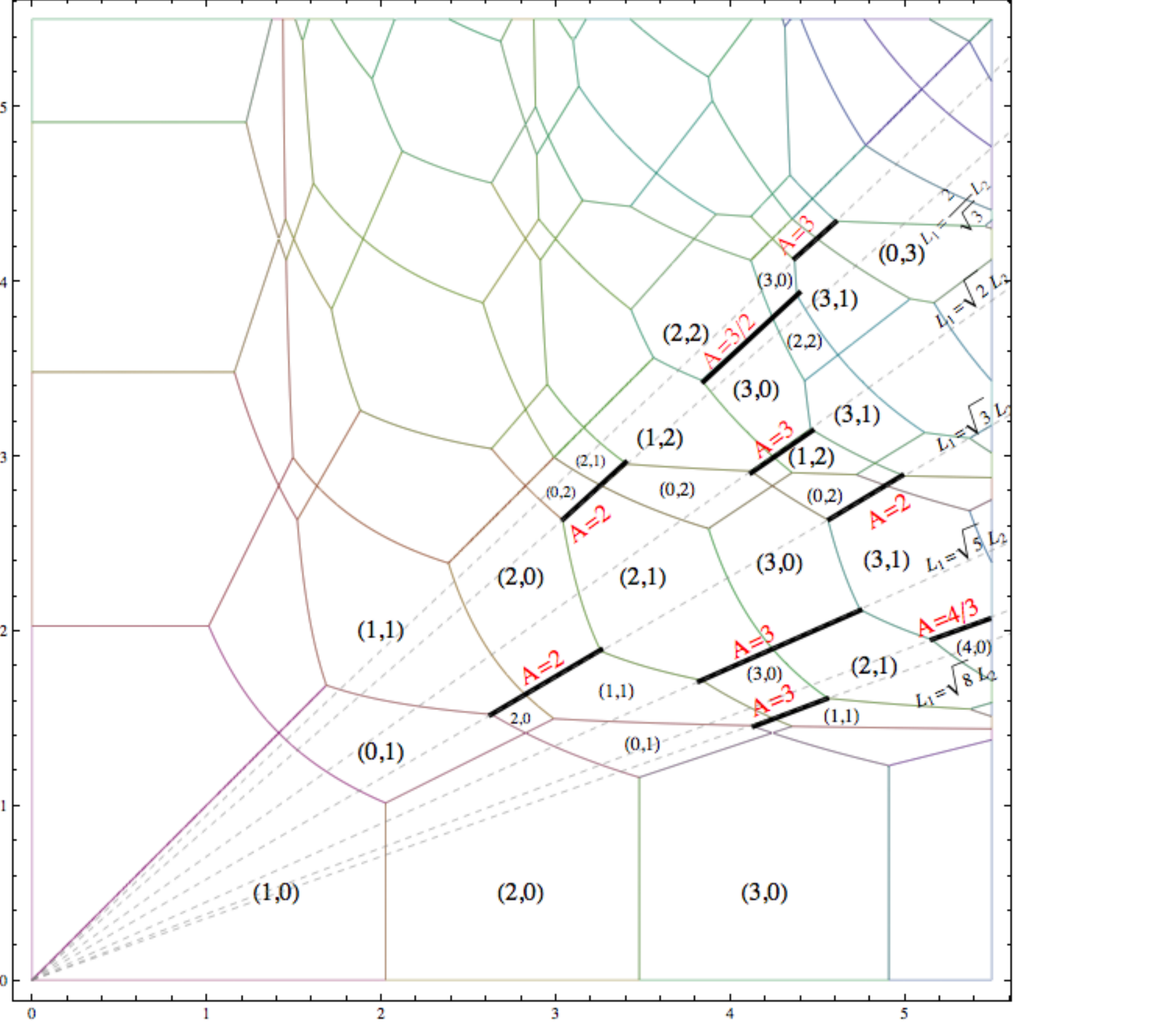}
\caption{The selection of the horizontal wave indices $(i_x,i_y)$. \label{wavemap}}
\end{figure}

\begin{figure}
\includegraphics[scale=1]{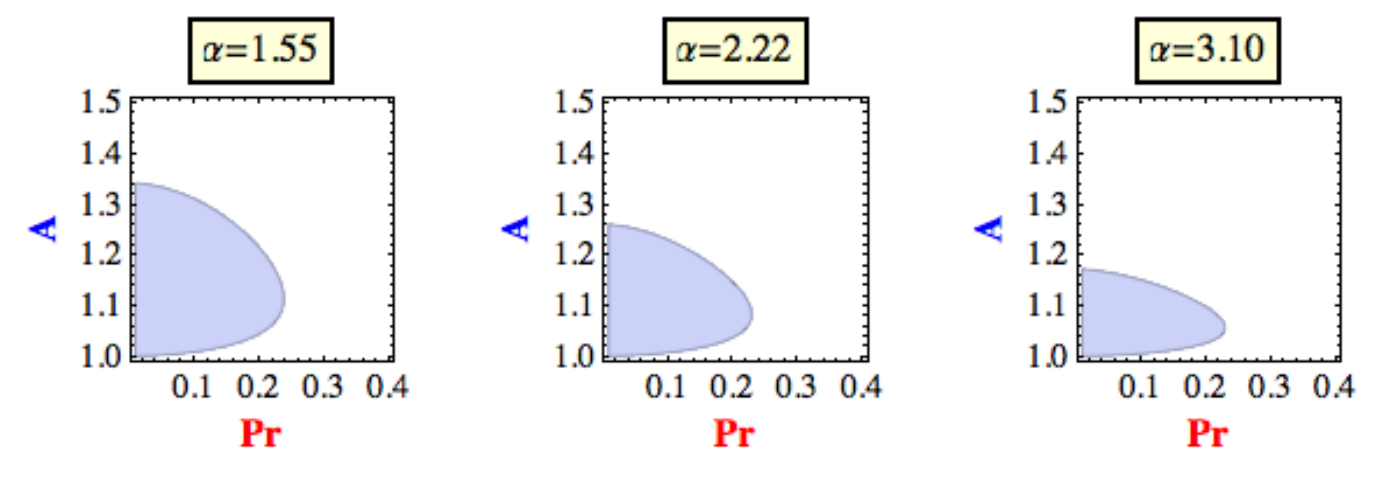}
\includegraphics[scale=1]{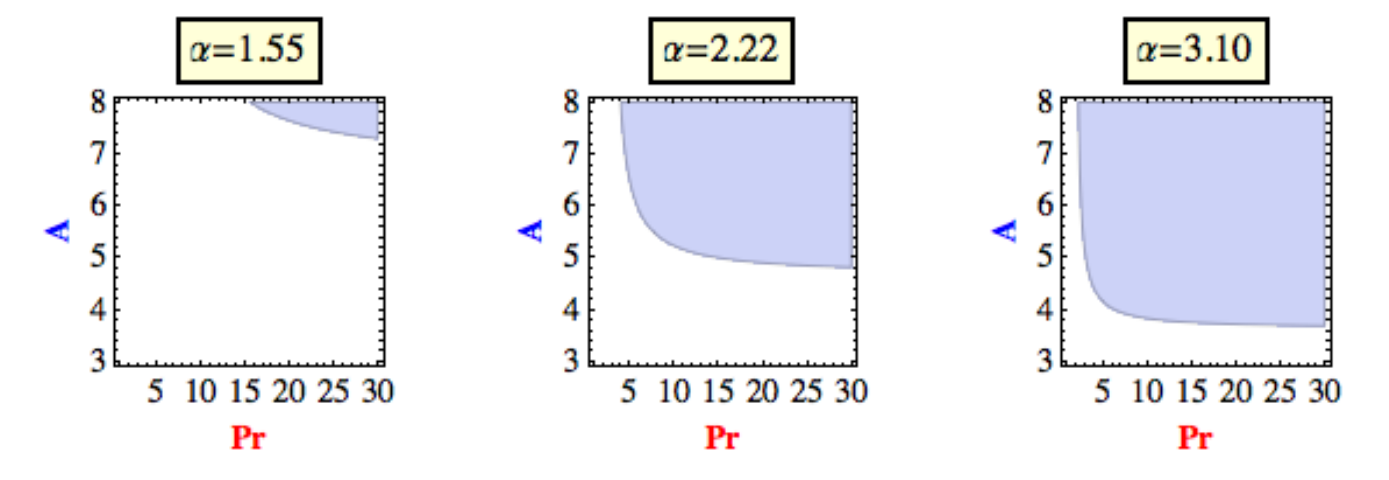}
\caption{The shaded regions show the parameter regimes where $c<a$ when two critical indices are $I=(i_x,i_y,1)$ and $J=(0,j_y,1)$. Here $A=j_y/i_y$. \label{acplot}}
\end{figure}

\begin{figure}
\includegraphics[scale=.58]{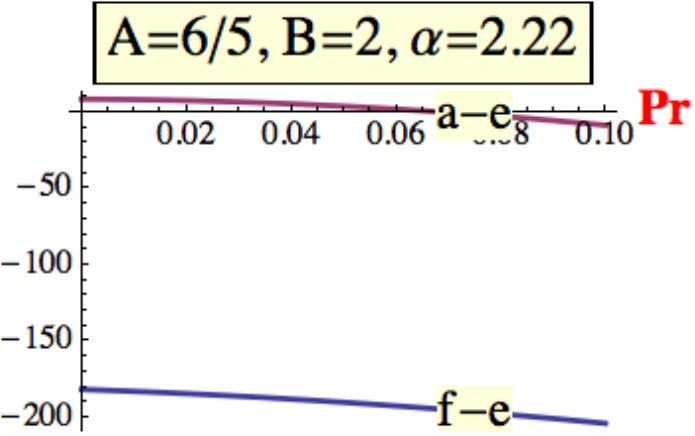}
\caption{The parameters $a-e$ and $f-e$ in the small Prandtl number regime. \label{recvsrectest}}
\end{figure}

\bibliographystyle{amsplain}
\bibliography{library}
\end{document}